\documentclass[a4paper]{article}

\usepackage[ansinew]{inputenc}
\usepackage{amsmath}
\usepackage{amssymb}
\usepackage{amsfonts}
\usepackage{tikz}
\usepackage{subfigure}
\usepackage{amsthm}

\usetikzlibrary{arrows}

\newcounter{saetning} 
  \newtheorem{theo}[saetning]{Theorem}
  \newtheorem{coro}[saetning]{Corollary}
  \newtheorem{lemm}[saetning]{Lemma}
  \newtheorem{prop}[saetning]{Proposition}
\theoremstyle{definition} 
	\newtheorem{defi}{Definition}
\theoremstyle{remark}
	
	\newtheorem{example}{Example}

\newcommand{\N}{\mathbb{N}}

\newcommand{\E}{\mathbb{E}}

\begin{document}

\markboth{Jakobsen et al.}{Timeability of Extensive-Form Games}

\title{Timeability of Extensive-Form Games}
\author{Sune K.~Jakobsen\thanks{School of Mathematical Sciences and School of Electronic Engineering \&
Computer Science, Queen Mary University of London, Mile End Road, London,
E1 4NS, UK. Email: S.K.Jakobsen@qmul.ac.uk.}
\and
Troels B.~S{\o}rensen
\thanks{IT-University of Copenhagen, Rued Langgaards Vej 7, 2300 Copenhagen S, Denmark. Email: trbj@itu.dk.}
\and
Vincent Conitzer
\thanks{Department of Computer Science, Duke University, 
Box 90129,
Durham, NC 27708, USA. Email: conitzer@cs.duke.edu.}
}

\maketitle 

\begin{abstract}
 Extensive-form games constitute the standard representation scheme
for games with a temporal component.  But do all extensive-form
games correspond to protocols that we can implement in the real world?
We often rule out games with {\em imperfect recall},
which prescribe that an agent forget something that she knew
before.  In this paper, we show that even some games with perfect
recall can be problematic to implement.  Specifically,
we show that if the agents have a sense of time passing
(say, access to a clock), then some extensive-form games
can no longer be implemented; no matter how we attempt to time
the game, some information will leak to the agents that they are not
supposed to have.  We say such a game is not
{\em exactly timeable}. We provide easy-to-check necessary and
sufficient conditions for a game to be exactly timeable.
Most of the technical depth of the paper concerns how to
{\em approximately} time games, which we show can always be done,
though it may require large amounts of time. Specifically, we show that for some games the time required to approximately implement the game grows as a power tower of height proportional to the number of players and with a parameter that measures the precision of the approximation at the top of the power tower. In practice, that makes the games untimeable.
Besides the conceptual contribution to game theory, we believe our
methodology can have applications
to preventing information leakage in security protocols.
\end{abstract}

\section{Introduction}

The {\em extensive form} is a very powerful representation scheme for
games.  It allows one to naturally
specify how the game unfolds over time, and what each player knows at
each point of action.
This allows one to model, for example, card games such as
poker, but also real-world
strategic situations with similar aspects.

Besides asking whether all strategic situations one might encounter in the 
real world
can be modelled as extensive-form games,
one may also ask whether all extensive-form games correspond to
something one might encounter in the real world.
This question is important for several reasons.  One is that if the
answer is ``no,'' then
there should be some well-motivated restricted subclasses of
extensive-form games that may be more tractable from the perspective
of algorithmic and other theoretical analysis.
Another is that if we are interested in designing a protocol,
extensive-form games give us a natural language
in which to express the protocol---but this language may lead us
astray if some of its games are not
actually implementable
in the real world.

Games of {\em imperfect recall}, in which an agent sometimes forgets
something she knew before, constitute a natural
example of games that may be difficult to implement in the
real world.\footnote{Computer poker provides some amusing
anecdotes in this regard. When comparing two poker-playing bots by
letting them play a sequence of hands,
one way to reduce the role of luck and thereby improve statistical
significance is to wipe clean the bots' memory
and let them play the same sequence of hands again, but with the bots'
roles in the hands reversed. This is not feasible
for {\em human} players, of course. Because of this, events pitting
computers against humans have generally pitted
a {\em pair} of players against one copy of the bot each, in separate
rooms. In this setup, each human-computer pair
receives the same hands, though the bot's role in one room is the
human's role in the other.} Indeed, restricting attention to perfect recall is often useful for algorithmic
and other theoretical purposes.
From a theoretical perspective, perfect recall is required~\cite{kuhn1953extensive} for behavioral strategies to be as
expressive as mixed strategies. Perfect recall also allows for the use of the
sequence form~\cite{romanovsky1962reduction}, which allows linear optimization
techniques to be used for computing equilibria of two-person extensive-form
games~\cite{von1996efficient}. The sequence form can also be used to compute
equilibrium refinements~\cite{miltersen2010computing,miltersen2008fast}, again
requiring perfect recall. Without perfect recall, otherwise simple single agent decision
problems become complicated~\cite{piccione1997interpretation,aumann1997absent,aumann1997forgetful},
and even the existence of equilibria in behavior strategies becomes NP-hard to
decide~\cite{hansen2007finding}. Imperfect recall has proven
useful for computing approximate minimax strategies for
poker~\cite{waugh2009practical}, even though the agent following the strategy does have perfect recall when playing the game.

We believe that many researchers are under the impression
that, given any finite extensive-form game
of perfect recall, one could
in principle have agents play that game in the real world, with the
actions of the game unfolding
in the order suggested by the extensive form. %skj: reference to examples?
In this paper, we prove that this is not so, at least if agents have a
sense of {\em time}.
If the players have a sense of time, we show that some games cannot be 
implemented in actual time in a way that respects the information
sets of the extensive form. 
The games that can be implemented in time are exactly those that have {\em 
chronologically ordered information sets}, as
defined in a set of lecture notes by~\cite[page 91]{weibull2009lecturenotes}.
Weibull argues that games with this property constitute the natural domain of {\em 
sequential equilibria}~\cite{kreps1982sequential}. The concept of sequential equilibrium is
arguably the most used equilibrium refinement for extensive-form games with 
imperfect information. ~\cite{kreps1987structural}
provided an example where the unique sequential equilibrium requires some level of cognitive 
dissonance from the players~\cite{weibull1992self}, forcing a player to best-respond to strategies
that are not consistent with her beliefs. However, examples of 
this type only work because they have no ordering of the information sets, 
which is Weibull's point in restricting attention to games with 
chronologically ordered information sets.
In this paper, we argue something stronger: we argue that extensive-form games 
without this property
cannot model any real world strategic situation, since the information 
structure of the model cannot be enforced.

We emphasize that our paper is not intended as a criticism of
extensive-form games. Rather, the goal is to point out a natural
restriction -- timeability -- that is needed to ensure that the game
can be implemented as intended in practice. Again, perfect recall is
a restriction that is similar in nature. Restricting attention to
those games that have perfect recall has been useful for many
purposes, and the notion has also been useful to understand why
certain games have odd features---namely, they have imperfect recall.  
We suspect the notion of timeability can be used similarly, and
encourage game theorists (algorithmic or otherwise) to, in contexts
where they consider the restriction of perfect recall, consider that
of timeability as well.  

One place where the analogy between timeability and perfect recall
perhaps breaks down is that we have shown that games that are not
exactly timeable can nevertheless be approximately timed, in some
cases even in a reasonable amount of time. It is not clear whether an
analogous notion of approximately perfect recall could be given.

Most of our
technical work concerns whether games that do not have an exact
timing can nevertheless be approximately timed, and if so, how much
time is required to do so. This latter contribution may have important
ramifications for the design of protocols that run a risk of leaking
information to participants based on the times at which
they are requested to take action. While we show that all games are at least 
approximately timeable, we also show that some games require so much time that 
in practice they are untimeable.

\subsection{Motivating example}

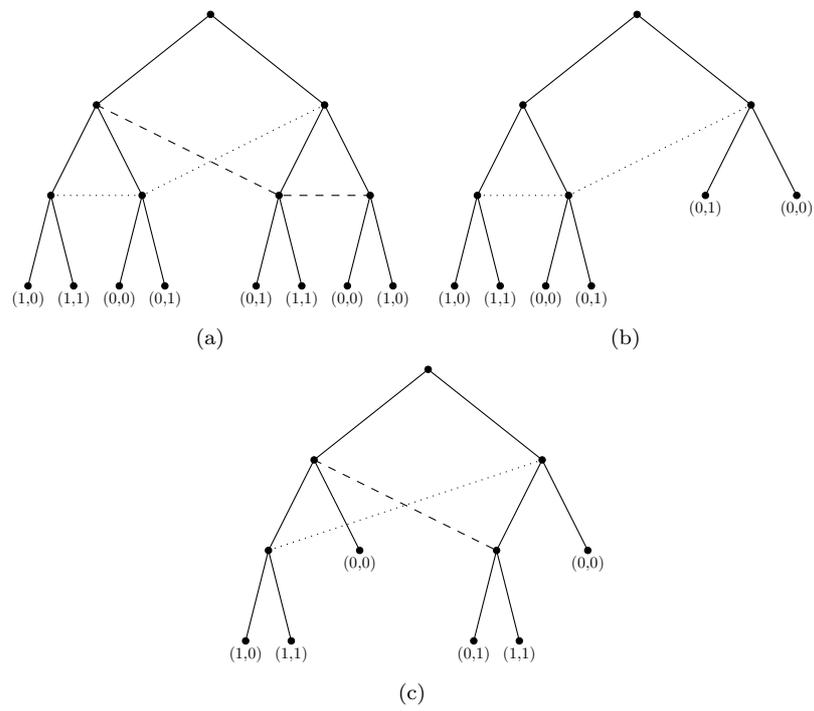
\begin{figure}
\centering

%\minipage{0.32\linewidth}
\subfigure[]{%
\begin{tikzpicture}[scale=0.6, every node/.style={scale=0.6}]

\draw (5.5,6) -- (3,4) -- (4,2) -- (4.5,0);
\draw[dashed]  (3,4) -- (7,2) -- (9,2); 
\draw[dotted]  (2,2) -- (4,2) -- (8,4); 
\draw (4,2) -- (3.5,0);
\draw (3,4) -- (2,2) -- (2.5,0);
\draw (2,2) -- (1.5,0);
\draw (5.5,6) -- (8,4) -- (9,2) -- (9.5,0);
\draw (9,2) -- (8.5,0);
\draw (8,4) -- (7,2) -- (7.5,0);
\draw (7,2) -- (6.5,0);

\draw [fill] (5.5,6) circle [radius=0.07];
\draw [fill] (3,4) circle [radius=0.07];
\draw [fill] (4,2) circle [radius=0.07];
\draw [fill] (4.5,0) circle [radius=0.07];
\node [below] at (4.5,0) {(0,1)}; 
\draw [fill] (3.5,0) circle [radius=0.07];
\node [below] at (3.5,0) {(0,0)}; 
\draw [fill] (2,2) circle [radius=0.07];
\draw [fill] (2.5,0) circle [radius=0.07];
\node [below] at (2.5,0) {(1,1)}; 

\draw [fill] (1.5,0) circle [radius=0.07];
\node [below] at (1.5,0) {(1,0)}; 
\draw [fill] (8,4) circle [radius=0.07];
\draw [fill] (9,2) circle [radius=0.07];
\draw [fill] (7,2) circle [radius=0.07];
\draw [fill] (9.5,0) circle [radius=0.07];
\node [below] at (9.5,0) {(1,0)}; 
\draw [fill] (8.5,0) circle [radius=0.07];
\node [below] at (8.5,0) {(0,0)}; 
\draw [fill] (7.5,0) circle [radius=0.07];
\node [below] at (7.5,0) {(1,1)}; 
\draw [fill] (6.5,0) circle [radius=0.07];
\node [below] at (6.5,0) {(0,1)}; 

\end{tikzpicture}

\label{fi:motivating}}
\subfigure[]{%
\begin{tikzpicture}[scale=0.6, every node/.style={scale=0.6}]

\draw (5.5,6) -- (3,4) -- (4,2) -- (4.5,0);
%\draw[dashed]  (3,4) -- (7,2) -- (9,2); 
\draw[dotted]  (2,2) -- (4,2) -- (8,4); 
\draw (4,2) -- (3.5,0);
\draw (3,4) -- (2,2) -- (2.5,0);
\draw (2,2) -- (1.5,0);
\draw (5.5,6) -- (8,4) -- (9,2);
\draw (8,4) -- (7,2);

\draw [fill] (5.5,6) circle [radius=0.07];
\draw [fill] (3,4) circle [radius=0.07];
\draw [fill] (4,2) circle [radius=0.07];
\draw [fill] (4.5,0) circle [radius=0.07];
\node [below] at (4.5,0) {(0,1)}; 
\draw [fill] (3.5,0) circle [radius=0.07];
\node [below] at (3.5,0) {(0,0)}; 
\draw [fill] (2,2) circle [radius=0.07];
\draw [fill] (2.5,0) circle [radius=0.07];
\node [below] at (2.5,0) {(1,1)}; 

\draw [fill] (1.5,0) circle [radius=0.07];
\node [below] at (1.5,0) {(1,0)}; 
\draw [fill] (8,4) circle [radius=0.07];
\draw [fill] (9,2) circle [radius=0.07];
\node [below] at (9,2) {(0,0)}; 
\draw [fill] (7,2) circle [radius=0.07];
\node [below] at (7,2) {(0,1)};

\end{tikzpicture}
\label{fi:modified}
}
\subfigure[]{
\begin{tikzpicture}[scale=0.6, every node/.style={scale=0.6}]

\draw (5.5,6) -- (3,4) -- (4,2) ;
\draw[dashed]  (3,4) -- (7,2) ; 
\draw[dotted]  (2,2) -- (8,4); 

\draw (3,4) -- (2,2) -- (2.5,0);
\draw (2,2) -- (1.5,0);
\draw (5.5,6) -- (8,4) -- (9,2) ;

\draw (8,4) -- (7,2) -- (7.5,0);
\draw (7,2) -- (6.5,0);

\draw [fill] (5.5,6) circle [radius=0.07];
\draw [fill] (3,4) circle [radius=0.07];
\draw [fill] (4,2) circle [radius=0.07];
\node [below] at (4,2) {(0,0)}; 
\draw [fill] (2,2) circle [radius=0.07];
\draw [fill] (2.5,0) circle [radius=0.07];
\node [below] at (2.5,0) {(1,1)}; 

\draw [fill] (1.5,0) circle [radius=0.07];
\node [below] at (1.5,0) {(1,0)}; 
\draw [fill] (8,4) circle [radius=0.07];
\draw [fill] (9,2) circle [radius=0.07];
\draw [fill] (7,2) circle [radius=0.07];
\node [below] at (9,2) {(0,0)}; 
\draw [fill] (7.5,0) circle [radius=0.07];
\node [below] at (7.5,0) {(1,1)}; 
\draw [fill] (6.5,0) circle [radius=0.07];
\node [below] at (6.5,0) {(0,1)}; 

\end{tikzpicture}
\label{fi:modified2}
}

\caption{Three examples. The roots are Chance nodes where Chance chooses its move uniformly at random. Dashed information sets belong to player $1$ and dotted ones to player $2$. The node in game (b) that forms its own information set belongs to player $1$.}
\end{figure}

Consider the following simple 2-player extensive-form game
(Figure~\ref{fi:motivating}).
In it, first a coin is tossed that determines which player goes first.
Then, each player, in turn, is asked
to guess whether she has gone first.
If the player is correct, she is paid $1$ (and otherwise $0$).
The information sets of the game suggest that a player cannot at all
distinguish the situation
where she goes first from the one where she goes second, and thus, she
gets expected utility $1/2$ no matter her strategy.

However, now consider implementing this game in practice.  Assume that
the game starts at time $0$.  Clearly,
if we toss the coin at time $0$, ask one player to bet at time $1$,
and the other at time $2$, a time-aware
player will know exactly whether she is being asked first or second
(assuming the timing protocol is common knowledge),
and will act accordingly.
This implementation blatantly violates the intended information
structure of the extensive-form representation of
the game; indeed, it results in an entirely different game (one that
is much more beneficial to the players!).
We say that this protocol is not an {\em exact timing} of the game in
Figure~\ref{fi:motivating}.

Of course, the general protocol of taking one action per time unit is
a perfectly fine timing of many games,
including games where every action is public (as in, say, Texas Hold'em poker).
Also, there are games where taking one action per time unit fails to
exactly time the game,
but nevertheless an exact timing is available.  For example, consider the
modified game in Figure~\ref{fi:modified},
where player 1 only plays if the coin comes up Heads, and if so plays
first.  This game can be timed
by letting player $1$ play at time $1$ and player $2$ at time $2$,
even if player $1$ does not go first.

But what about the game in Figure~\ref{fi:motivating}?  Can it not be
timed at all?
We will pose the constraint that there must be at least one time unit
between successive actions in the extensive form.
Without this constraint, we could take the normal form of the game and
let players play it by declaring their entire
strategy at once---but this scheme violates the natural interpretation
of the extensive form, and would allow us
to play games of imperfect recall just as well.  (One may argue that
we should just let the players play in parallel after
the coin
flip
in the game in Figure~\ref{fi:motivating}---however, a simple
modification of the game where the second player
is only offered a bet if the first player guessed correctly
(Figure~\ref{fi:modified2}) would disallow this move.)
It is easy to see that
no {\em deterministic} timing will suffice.  This is because every
node within an information set would have to have
the same time associated with it; but then, the left-hand side of the
tree requires that player $1$'s information set
has a time strictly before that of player $2$, but the right-hand side
implies the opposite.

For games where deterministic timing cannot be done, one might turn to 
randomized timing when trying to implement the game. However, if
%Randomization will not completely get us out of trouble, either.
the time at which a node is played is
to reveal {\em no information whatsoever} about which node in the
information set has been reached, then the
{\em distribution} over times at which it is played must be identical for
each node in the information set.  But this
cannot be achieved in the game in Figure~\ref{fi:motivating}, because
the left-hand side of the tree ensures that
the expectation of the time distribution for player $1$'s information
set must be at least $1$ lower than that
for player $2$'s information set, but the right-hand side implies the opposite.
Still, we may achieve {\em something} with randomization.
For example, we may draw an integer $i$ uniformly at random from 
$[N-1]=\{1, \ldots, N-1\}$,
offer the first player a bet at time $i$
and the second player a bet at time $i+1$. Then, if a player is offered a bet
at time $1$ or time $N$, the player will know exactly
at which node in the extensive form she is.  On the other hand, if she
is offered a bet at any time
$t \in \{2, \ldots, N-1\}$, she obtains
no additional information at all, because the conditional probability
of $t$ being the selected time is
the same whether she is the first or the second player to move.
Hence, as long as $i \in \{2, \ldots, N-2\}$, which happens with
probability $(N-3)/(N-1)$,
neither player learns anything from the timing.
We say the game is {\em approximately timeable}: we can come
arbitrarily close to timing the game by increasing
$N$, the number of time periods used.
This immediately raises the question of whether {\em all} games are
approximately timeable,
and if so how large $N$ needs to be for a particular approximation.

\subsection{Our contribution}
In the next section we define \emph{exactly timeable games}, give a characterization of these games, and show that there is a linear-time algorithm that decides whether an extensive-form game is exactly timeable. In Section \ref{sec:epstime} we define $\epsilon$-timeability and argue that this is the correct definition. In Section \ref{sec:upper} we show that all extensive-form games are $\epsilon$-timeable for any $\epsilon>0$, but in Section \ref{sec:lower} we show that these $\epsilon$-timings can easily become too time-consuming for this universe: for any number $r$, there exists a game $\Gamma_r$ such that for sufficiently small $\epsilon$, any $\epsilon$-timing of $\Gamma_r$ will take time at least $2^{2^{\dots 2^{\frac{1}{\epsilon}}}}$ where the tower has height $r$. In Section \ref{sec:imperfect} we ask what happens if we have some control over the players' perception of time. We assume that there exists a constant $c$ such that any player will always perceive a time interval of length $t$ as having length between $\frac{t}{c}$ and $ct$, and otherwise we have complete control over the players' perception of time. We show that even under these assumptions, the lower bound from Section \ref{sec:lower} still holds.

\section{Exactly timeable games}
\begin{defi}
For an extensive-form game\footnote{For an introduction to the game-theoretical concepts used in this paper, see, for example,~\cite{AGT07}} $\Gamma$, a \emph{deterministic timing} is a labelling of the nodes in $\Gamma$ with non-negative real numbers such that the label of any node is at least one higher than the label of its parent. A deterministic timing is \emph{exact} if any two nodes in the same information set have the same label.
\end{defi}
An exact deterministic timing is the same as the time function in the
definition of a chronological order by~\cite{weibull2009lecturenotes}.
%Weibull~\cite{weibull2009lecturenotes}.
Since we will also be discussing games that cannot be timed, we need this more
general definition of timings that are not exact.

This definition allows times to be nonnegative real numbers rather than 
integers, which makes some of the proofs cleaner. However, given a deterministic timing with real values, one can always 
turn it into a timing with integer values by taking the floor function of each 
of the times. 
  %vc: should we cite/discuss Weibull for this definition?
 %tbs: something like this?

  The following theorem says that it is easy to check whether a game has an exact deterministic timing, providing multiple equivalent criteria. Criterion $2$ is presumably most useful for a human being looking at small extensive-form games, while criterion $3$ is easy for a computer to check.
  
\begin{theo}\label{theo:char}
For an extensive-form game $\Gamma$, the following are equivalent:
\begin{enumerate}
\item $\Gamma$ has an exact deterministic timing.
\item The game tree $\Gamma$ can be drawn in such a way that a node always has a lower $y$-coordinate than its parent, and two nodes belong to the same information set if and only if they have the same $y$-coordinate.
\item Contracting each information set in the directed graph $\Gamma$ to a single node results in a graph without oriented cycles.
\end{enumerate}
\end{theo}
\begin{proof}
``$1\Rightarrow 2$:'' Given an exact deterministic timing 
(WLOG, with integer-valued times), we draw $\Gamma$ such that each node has $y$-coordinate equal to the negative of its time. As the timing is exact, nodes in the same information set have the same $y$-coordinate. To ensure that any two nodes with the same $y$-coordinate are in the same information set, we perturb each node based on its information set. 
This can be done deterministically: for example, if there are $q$ information sets in the game, then subtract $i/q$ from the time of each node in the $i$th information set.

``$2\Rightarrow 3$:'' Given such a drawing, contracting each information set results in all edges going downwards, so the resulting graph cannot have directed cycles.

``$3\Rightarrow 1$:'' The nodes of a directed acyclic graph can be numbered such that each edge goes from a smaller to a larger number. This numbering can be used as a deterministic timing.\end{proof}
  
  We can use criterion $3$ of Theorem \ref{theo:char} to test whether the games in 
Figure \ref{fi:motivating} and \ref{fi:modified} are timeable. First we draw 
a node for each information set: One for the root, one for player $1$'s 
information set and one for player $2$'s information set. (If one of the 
players had more than one information set, that player would have had more than one 
node in the contracted graph.) We ignore the leaves, as they can never form 
cycles. 
In the games in Figure \ref{fi:motivating} and \ref{fi:modified}, we 
can get from the root to each of the two players' information sets, so we draw 
a directed edge from the root to each of the two other nodes. We can 
also get from player $1$'s information set to player $2$'s, and in the game 
in Figure \ref{fi:motivating} we can go from player $2$'s information set to 
player $1$'s. When we draw these directed edges (without multiplicity) we get 
Figure \ref{fi:motivatingcon} and Figure \ref{fi:modifiedcon}, respectively. 
We see that the graph in Figure \ref{fi:motivatingcon} has a cycle, so the
game in Figure  \ref{fi:motivating} is not exactly timeable, while graph in Figure 
\ref{fi:modifiedcon} does not have a cycle, so the game in Figure 
\ref{fi:modified} is exactly timeable. 
  The contracted graph can be constructed in linear time, and given this directed 
graph, we can in linear time test for cycles~\cite{CLRS}. Thus, we can test 
in linear time whether a game is exactly timeable.

\begin{figure}
\centering

%\minipage{0.32\linewidth}
\subfigure[]{
\begin{tikzpicture}

\draw [-triangle 60] (2,3) -- (1,1.5);
\draw (1,1.5)--(0,0); 
\draw [-triangle 60] (2,3) -- (3,1.5);
\draw (3,1.5) -- (4,0);
\draw [-triangle 60] (0,0) to [out=30,in=180] (2.1,0.5);
\draw (2.1,0.5) to [out=0,in=150] (4,0);
%\draw [-triangle 60] (0,0) to [out=30,in=150] (4,0);
\draw [-triangle 60] (4,0) to [out=210, in=0] (1.9,-0.5);
\draw (1.9,-0.5) to [out=180, in=330] (0,0);

\draw [fill] (0,0) circle [radius=0.05];
\draw [fill] (4,0) circle [radius=0.05];
\draw [fill] (2,3) circle [radius=0.05];

\end{tikzpicture}

\label{fi:motivatingcon}}
\subfigure[]{
\begin{tikzpicture}

\draw [-triangle 60] (2,3) -- (1,1.5);
\draw (1,1.5)--(0,0); 
\draw [-triangle 60] (2,3) -- (3,1.5);
\draw (3,1.5) -- (4,0);
\draw [-triangle 60] (0,0) to (2.1,0);
\draw (2.1,0) to (4,0);

\draw [fill] (0,0) circle [radius=0.05];
\draw [fill] (4,0) circle [radius=0.05];
\draw [fill] (2,3) circle [radius=0.05];

\draw [-triangle 60,white] (4,0) to [out=210, in=0] (1.9,-0.5);

\end{tikzpicture}
\label{fi:modifiedcon}
}

\caption{Example of how to use Theorem \ref{theo:char}. The top node is the Chance node, the left node corresponds to player $1$'s information set, and the right node corresponds to player $2$'s information set.}
\end{figure}
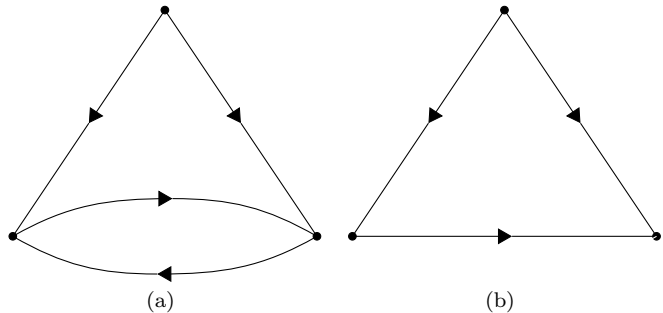

  %vc: why is there nothing about random exact timing in here anymore?  I know that doesn't help at all, but don't we want to show that? ... oh I see, that's later, OK.
  %skj: because we only give the definition later. Otherwise it should really be part of the theorem.

  \section{$\epsilon$-timeability}\label{sec:epstime}
  
  We now move on to approximate timeability.    
%vc: a bit ugly to have this massive definition instead of having it more split up
%skj: not sure how much you want to split it up, but you are welcome to do it.
  \begin{defi}
  The \emph{total variation distance} (also called \emph{statistical distance}) between two discrete random variables $X_1$ and $X_2$ is given by
  \[\delta(X_1,X_2)=\sum_{x}\max(\Pr(X_1=x)-\Pr(X_2=x),0)\]
  where the sum is over all possible values of $X_1$ and $X_2$. This measure is symmetric in $X_1$ and $X_2$.
  If $\delta(X_1,X_2)\leq \epsilon$ we say that $X_1$ and $X_2$ are \emph{$\epsilon$-indistinguishable}.
  
 A \emph{(randomized) timing} is a discrete distribution over deterministic timings.  
 For a game, a timing of the game, a player and a node $v$ belonging to that player, the player's \emph{timing information} is the sequence of times for nodes belonging to that player on the path from the root to $v$ (including $v$ itself). Thus, for a fixed game, timing, player, and node, the timing information is a random variable.

  The timing is an \emph{$\epsilon$-timing} if for any two nodes in the same information set, the total variation distance between the timing information at the two nodes is less than $\epsilon$. A $0$-timing is also called an \emph{exact timing}.
 
  A game is \emph{exactly timeable} if it has an exact timing,
   \emph{$\epsilon$-timeable} if it has an $\epsilon$-timing, and \emph{approximately timeable} if it is $\epsilon$-timeable for all $\epsilon>0$.
  \end{defi}
  
  %vc: maybe the following should be a proposition or observation or something.
  %skj: I called it a theorem because I considered it to be part of Theorem 1. Of course I should have indicated this. Could we called it something like "Theorem 1 continued"? I agree that otherwise we should change it to a proposition. ... I have now added the line below and made it a proposition instead of a theorem.
  The following proposition implies that $\Gamma$ being exactly timeable is equivalent to each of the three criteria in Theorem \ref{theo:char}.
  \begin{prop}
  A game is exactly timeable if and only if it has an exact deterministic timing.
  \end{prop}
  \begin{proof}
  An exact deterministic timing is a special case of an exact randomized timing. Conversely, given an exact randomized timing of a game, we can label each node with its expected time to obtain an exact deterministic timing.
  \end{proof}
  
 We will show that in fact all games are approximately timeable. 
We will need the following properties of total variation distance. For proofs, see~\cite{NoF}.

%\begin{prop}[Triangle inequality]\label{prop:triangle}
%For random variables $X_1,X_2,X_3$ we have
%$\delta(X_1,X_3)\leq \delta(X_1,X_2)+\delta(X_2,X_3).$
%\end{prop}

\begin{prop}[Data Processing Inequality for Total Variation Distance]\label{prop:dpi}
Suppose $X_1$ and $X_2$ have total variation distance $\epsilon$, $Y$ is a random variable independent from $X_1$ and $X_2$, and $f$ is a function. Then the total variation distance between $f(X_1,Y)$ and $f(X_2,Y)$ is at most $\epsilon$.
\end{prop}

\begin{prop}\label{prop:disjoint}
Let $X_1,\dots, X_n,Y_1,\dots ,Y_n,I$ be independent random variables with $X_i$ and $Y_i$ distributed on $\mathcal{X}_i$, and $I$ distributed on $[n]$. Let $X=X_I$ and $Y=Y_I$. We have
$\delta(X,Y)\leq \sum_{i=1}^n\Pr(I=i)\delta(X_i,Y_i)$,
with equality if all the $\mathcal{X}_i$'s are pairwise disjoint.
\end{prop}

\begin{prop}\label{prop:tooln2}
If $X_1$ and $X_2$ are discrete random variables taking real values in an interval $[a,b]$ and $\E X_2\geq \E X_1+1$ then 
$\delta(X_1,X_2)\geq \frac{1}{b-a}.$
\end{prop}

\begin{prop}\label{prop:givenT}
Let $(X_1,X_2,T)$ be random variables with some joint distribution, where $T$ only takes values $0$ and $1$, $\Pr(T=0)=\epsilon<1$, and $\delta(X_1,X_2)=\delta$. For $i\in\{1,2\}$ define $X_i'=X_i|_{T=1}$. Then 
$\delta(X_1',X_2')\leq \frac{\delta+\epsilon}{1-\epsilon}.$
\end{prop}

\subsection{$\epsilon$-timeability and $\epsilon$-approximate Nash equilibria}

In this section we will argue that our definition of $\epsilon$-timeability is the correct one, because it gives a bound on how much the players can gain in expectation per choice from having the timing information. If one is playing a game $\Gamma$ with timing $X$, one can think of this as really playing a different game $\Gamma'$ whose first node is a Chance node that chooses from among all the possible timings, with a distribution given by $X$. Then, each possible choice by Chance leads to a copy of $\Gamma$, and two nodes in $\Gamma'$ are in the same information set if and only if they correspond to nodes in $\Gamma$ in the same information set and have the same timing information.

We want to show that if $\Gamma$ has utilities in $[0,1]$, $X$ is an $\epsilon$-timing, and a player has at most $m$ nodes in any history, then that player can gain at most $m\epsilon$ in expectation from using the timing information.\footnote{We do not rely on anything specific to timing here: the theorems and examples in this section hold
 for any side information such that the total variation distance between the side information for any two nodes in the same information set is at most $\epsilon$.} First, we show that there are games where the gain can get arbitrarily close to $m\epsilon$.

\begin{example}
Consider a game with $m$ rounds and only one player. In each round, Chance chooses a number uniformly at random from $[k]$ and then the player can either try to guess the number or pass. If she guesses, the game ends in that round and she gets utility $1$ if she was correct and $0$ otherwise. If she passes, the player learns the value that Chance chose, and they play another round. In the last round, she is not allowed to pass. 

It is clear that without any further information the player will get utility exactly $\frac{1}{k}$ no matter what strategy she uses. Now consider the following timing: each node usually happens $1$ later than the previous one. The only exceptions are the nodes belonging to the player. At each of these nodes there is probability $\epsilon$ that it is delayed (chosen independently for each node), and if it is delayed, it will be delayed with $i$ time units, where $i$ is the number chosen by Chance. 

%skj: I have rewritten the paragraph below. (I also added "(chosen independently for each node)" in the above paragraph)

First, we check that this is an $\epsilon$-timing. We choose two nodes $v$
and $w$ in the same information set. As they are in the same information
set, they would have to be in the same round $j$. Let
$X_v=(X_{v,1},\dots,X_{v,j})$ denote the timing information at $v$ and
similar for $w$. Because the player learns Chance's choice after she
passes, two nodes only belong to the same information set if Chance made
the same choices in all previous rounds. Thus, $X_{v,i}=X_{w,i}$ for all
$i<j$, and the only difference between $X_v$ and $X_w$ comes from the
difference between $X_{v,j}$ and $X_{w,j}$.  Then, for fixed
$(x_1,\dots,x_{j-1})$ there is a $t$ such that $\Pr(X_{v,j}=t|(X_{v,1}\dots
X_{v,j-1})=(x_1,\dots x_{j-1}))=1-\epsilon=\Pr(X_{w,j}=t|(X_{v,1}\dots
X_{v,j-1})=(x_1,\dots,x_{j-1}))$.  Thus, $\delta(X_v|_{(X_{v,1},\dots,
  X_{v,j-1}=(x_1,\dots,x_{j-1})},X_{w}|_{(X_{v,1},\dots,
  X_{v,j-1})=(x_1,\dots x_{j-1})})\leq \epsilon$.  By averaging over
possible values of $(x_1,\dots,x_{j-1})$ we get $\delta(X_v,X_w)\leq
\epsilon$, so the timing we have defined is an $\epsilon$-timing.

We can now define a strategy for the player: if she sees a delay of time $i$ she will guess that Chance chose $i$. If she does not see a delay she will pass except in the last round where she will guess randomly. The probability that she does not see a delay is $(1-\epsilon)^m$. Given that she does see a delay she is correct with probability $1$ and otherwise with probability $\frac{1}{k}$. As $(1-\epsilon)^m\approx 1-m\epsilon$ for small $\epsilon$, her expected utility is 
\[(1-(1-\epsilon)^m)+\frac{(1-\epsilon)^m}{k}\approx m\epsilon \frac{1-m\epsilon}{k}=\frac{1+(k-1)m\epsilon}{k}.\]
Thus, for small $\epsilon$, her advantage is close to $\frac{k-1}{k}m\epsilon$, and for large $k$ this is close to $m\epsilon$. 
\end{example}

Next we will see that one cannot gain more than $m\epsilon$ by using the timing information. First we show it for a one-player game with only one information set.

\begin{lemm}\label{lemm:fakeinfo}
Let $\Gamma$ be a one-player game where each history only contains two moves. First, Chance makes a move, and then the player. All the player's nodes belong to the same information set. All utilities are in $[0,1]$. Let $\Gamma'$ be the game $\Gamma$ with an $\epsilon$-timing $X$. For any strategy $\sigma'$ for $\Gamma'$ achieving expected utility $u$ we can construct a strategy $\sigma$ for $\Gamma$ achieving utility at least $u-\epsilon$. This $\sigma$ only depends on $\sigma'$ and the timing, not on Chance's probabilities, the utilities at each node, or $u$. 
\end{lemm}
\begin{proof}
Let $\{1,\dots, k\}$ denote the nodes belonging to the player, and let $X_i$ denote the timing information available in $\Gamma'$ at node $i$. The strategy $\sigma'$ can be given as a function that, given timing information $x$ and some randomness $Y'$, gives an action $\sigma'(x,y')$; then, $\sigma$ should be a function that only takes randomness $Y$ and still outputs an action. We can decide the distribution of the randomness, so we take $Y=(X_1,Y')$ where $X_1$ and $Y'$ are independent. That is, the player generates her own timing information as if she were in node $1$. 

As $\Gamma'$ is $\epsilon$-timed we have $\delta(X_1,X_i)\leq \epsilon$. By Proposition \ref{prop:dpi} this implies \begin{align*}
\delta(u(\sigma(Y)),u(\sigma'(X_i,Y')))&=\delta(u(\sigma'(X_1,Y')),u(\sigma'(X_i,Y)))\leq \epsilon,
\end{align*}
and as $u$ only takes values in $[0,1]$, Proposition \ref{prop:tooln2} implies that $\E u(\sigma'(X_i,Y'))- \E u(\sigma(Y))\leq \epsilon$. Thus, for each possible choice by Chance, the player loses at most $\epsilon$ in expected utility given that choice. 
\end{proof}

We can now extend to general extensive-form games.

\begin{theo}\label{theo:ignoretime}
Let $\Gamma$ be an extensive form game of perfect recall game with utilities in $[0,1]$ where player $i$ has at most $m$ nodes in any history, and let $\Gamma'$ be an $\epsilon$-timed version of $\Gamma$. If $\sigma'_i$ is a player $i$ strategy for $\Gamma'$ there is a strategy $\sigma_i$ that does not use timing information, such that for any strategy profile $(\sigma'_i,\sigma'_{-i})$ for $\Gamma'$ we have $u_i(\sigma'_i,\sigma'_{-i})-u_i(\sigma_i,\sigma'_{-i})\leq m\epsilon$.
\end{theo}
\begin{proof}
Fix $\Gamma$ and $\Gamma'$ and strategy profile $(\sigma'_i,\sigma'_{-i})$. For each information set in $\Gamma$ we modify player $i$'s strategy as in Lemma \ref{lemm:fakeinfo}. Notice that we do not need to know the other players' strategies to do this. For each information set, she generates new timing information even for the previous nodes for which she has already generated timing information before. 

For each node $v$ we now compute player $i$'s expected utility given that she is at that node. We know from Lemma  \ref{lemm:fakeinfo} that when we change the strategy at one node, the expected utility given that we reach that node goes down with at most $\epsilon$. We order the nodes belonging to player $i$ in layers: If there are $j$ nodes belonging to player $i$ on the path from the root to node $v$, we put node $v$ in the $j$'th layer. As we have perfect recall, each information set is contained in one layer, and we have at most $m$ layers. If we modify the strategy for all the information set in one layer at a time, the total excepted utility goes down with at most $\epsilon$ per layer, so if we modify the strategy for all layers, it goes down with at most $m\epsilon$.
%vc: we're using this lemma on the current node as if it's a single-round game?  could use a bit more explanation.
%skj: I have rewritten the paragraph, see the above paragraph (I deleted some "%vc:"'s about the paragraph I have now rewritten.)
\end{proof}

We then obtain the following corollary which justifies our definition of $\epsilon$-timeability.

\begin{coro}
Let $\Gamma$ be a perfect recall game with utilities in $[0,1]$ where each player has at most $m$ nodes in any history. If $\Gamma'$ is an $\epsilon$-timing of $\Gamma$ then any Nash-equilibrium $\sigma$ of $\Gamma$ is an $m\epsilon$-approximate Nash equilibrium of $\Gamma'$. 
\end{coro}

\section{Upper bound}\label{sec:upper}

From~\cite{NoF} we have the following definition and theorem.

\begin{defi}
Let $X_1,\dots,X_n$ be random variables over $\N$ with some joint distribution such that we always have $X_1<X_2<\dots <X_n$.
 We say that $(X_1,\dots, X_n)$ 
 \emph{has $\epsilon$-indistinguishable $m$-subsets} if for any two subsets $\{i_1,\dots,i_m\},\{j_1,\dots j_m\}\subset [n]$ of size $m$, the two random sets $\{X_{i_1},\dots ,X_{i_n}\}$ and $\{X_{j_1},\dots, X_{j_n}\}$ are $\epsilon$-indistinguishable. We slightly abuse notation and say that $(X_1,\dots, X_n)$ \emph{has $\epsilon$-indistinguishable subsets} if for all $m<n$ it has $\epsilon$-indistinguishable $m$-subsets.
\end{defi}

In the following $\exp_2$ denotes the function given by $\exp_2(x)=2^x$, and $\exp_2^n(x)$ denotes iteration of $\exp_2$, so $\exp_2^n(x)=2^{2^{\dots^{2^x}}}$ where the tower contains $n$ $2$'s.

\begin{theo}
%vc: should there be a for all epsilon at the beginning of this sentence?
%skj:fixed
For fixed $n$ there exists a function $N:(0,1]\to \N$ such that for all $\epsilon>0$ there exists a distribution of $(X_1,\dots, X_n)$ such that
$1\leq X_1<X_2<\dots<X_n\leq N(\epsilon)$ are all integers and $X$ has
$\epsilon$-indistinguishable subsets. We can choose $N$ such that 
$N(\epsilon)=\exp_2^{n-2}\left(O\left(\frac{1}{\epsilon}\right)\right)$. Conversely, for such a distribution to exist, we must have $N(\epsilon)=\exp_2^{n-2}\left(\Omega\left(\frac{1}{\epsilon}\right)\right)$ for sufficiently small $\epsilon$.
\end{theo} 

The following gives intuition for the upper bound on $N$.
For $n=2$ it is easy to construct $(X_1,X_2)$ that has
$\epsilon$-indistinguishable subsets. For example, we can take $X_1$ to be
uniformly distributed on $[N-k]$ for some constants $N$ and $k$ and set
$X_2=X_1+k$. We can then use a recursive construction for higher $n$. If
$(X_1,\dots, X_n)$ has $\epsilon$-indistinguishable subsets and consecutive
$X_i$'s are usually not too close to each other, we can construct $(Y_1,Y_2
\dots, Y_{n+1})$ that has $\epsilon'$-indistinguishable subsets for some
$\epsilon'$. To do this, we choose $Y_1$ uniformly between $1$ and a
sufficiently large number, and choose each gap $Y_{i+1}-Y_{i}$ uniformly
and independently from $\left[2^{X_i}\right]$. For a proof that this works,
see~\cite{NoF}.

The intuition about the lower bound on $N$ is that for $n=2$ an
$n-1$-subset contains $1$ number, and the \emph{size} of this number gives
away some information about whether it is the higher or lowest. For $n=3$
an $n-1$ subset contains two numbers and their \emph{distance} gives away
some information about whether it is the middle number or another number
that is missing from the set. For $n=4$ an $n-1$ subset contains $3$
numbers, and now the \emph{ratios between the two distances} gives away
information about which number is missing, and so on.

We can use the construction to approximately time any game. 

%vc: should we put the upper bound on time in this theorem?
%skj: done. Although it is not a good bound. We could probably let the high of the tower be only maximal number of nodes belonging to one player in any history, and have some small dependence on number of players 
\begin{theo}
All games with at most $m$ nodes in each history can be $\epsilon$-timed in time $\exp^{m-3}\left(O\left(\frac{1}{\epsilon}\right)\right)$. In particular, all games are approximately timeable.
\label{th:all_approximate}
\end{theo}
\begin{proof}
  Take any game and $\epsilon>0$. We want to show that the game is
  $\epsilon$-timeable. First we find some
  distribution of $(X_1,\dots,X_{m-1})$ that has $\epsilon$-indistinguishable
  subsets. Now we let the time of the root be $0$ and the time of a node at depth $d$ be given by
  $X_d$. As the $X_d$ take values in $\N$ and are increasing this gives a
  timing of the game. If two nodes $v$ and $w$ belong to the same information set, 
%vc: for such phrasings, I'm worried some readers will get confused and take a history to
%include the timing information.  maybe add something like "in the original
%extensive form" in such cases?
%skj: I changed histories to nodes. I also gave the root time 0 (i had to spilt into two cases) to lower the hight of the tower by one
 the player $i$ who owns these nodes will have the same number $j-1$ of previous nodes at $v$ and at $w$. 
  As $(X_1,\dots, X_{m-1})$ has $\epsilon$-indistinguishable
  subsets, it has $\epsilon$-indistinguishable $j$-subsets, so if the root does not belong to player $i$ there is total variation distance less at most $\epsilon$ between the two nodes'
  timing information. Similarly, $(X_1,\dots, X_{m-1})$ has $\epsilon$-indistinguishable $j-1$-subsets, so if the root belongs to player $i$ the total variation distance between the two nodes' timing information is also at most $\epsilon$. 
\end{proof}

\section{Lower bound for timing of games}\label{sec:lower}

The downside of the positive result in Theorem~\ref{th:all_approximate} is
that the amount of time needed is astronomical.  
%skj:this is not really correct. I think in astronomy you usually use towers of hight two (or at most 3 if taking about the universe reaching thermal equilibrium). Can we say "beyond astronomical"?
We next show that this is
to some extent inevitable.  Our goal is to show that for any $r$ there
exist games that cannot be $\epsilon$-timed in time
$\exp_2^r\left(O\left(\frac{1}{\epsilon}\right)\right)$. As we only aim to
show the existence of such games, we do not have to think about all
possible games, but can rather concentrate on games that are easy to
analyze. In the following we will only look at \emph{choiceless}
games. These are extensive-form games that begin with a move by Chance and
in which all the following nodes only have one child. The players do not
learn Chance's move, so two nodes are in the same information set if they
belong to the same player and have the same number of previous nodes
belonging to that player. In these games the players never have a choice,
so from the standard game theory point of view, the fact that nodes in this
game belong to players is pointless. However, when we require the game to
be timed, the nodes play an important role: a player learns the time of all
her nodes in the history that is played, and this might reveal some
information about which branch of the tree the players are in. (Note also
that any choiceless game that takes a long time to $\epsilon$-time can be
turned into a game with choice that takes as long time to $\epsilon$-time,
simply letting each player make a move at each node that does not affect
future actions or information in the game.)

We will restrict our attention even further, to a class of games we will
call \emph{symmetric choiceless games}. These games are given by the number
of players $n$ and a finite sequence over $[n]$. In the first move Chance
chooses uniformly at random between all the $n!$ numberings of the $n$
players. After this move, the players have nodes in the order given by the
sequence. For example, if $n=3$ and the sequence is $233112$ then the first
node will belong to the player to whom Chance assigned the number $2$, the
next node will belong to the player to whom Chance assigned the number $3$,
and so on.

An instance of a timing of the $233112$ game would give a time to each of
the $6$ nodes for each of the $6$ histories (that is, for each possible
assignment of numbers $1,2,3$ to the three players). Thus, a deterministic
timing contains $36$ times, and a timing is then a random tuple of $36$
times. The timing is an $\epsilon$-timing if for any player and any two
histories, the timing information for the two histories have statistical distance at most
$\epsilon$. 
This implies that one player, Alice, 
%vc: should we insert "at the end of the game" here to prevent confusion?
%skj: actually it is at any point in the game. I added that, but you are welcome to change it
cannot, at any point in the game, tell (to a significant extent) the difference between when she got
assigned the number $1$ and when she got assigned the number $2$. However,
for the history where Alice is $1$ and Bob is $2$, the difference between
Alice's timing information and Bob's timing information is allowed to be
large because they are different players.

We say a timing of a symmetric choiceless game is \emph{symmetric} if the
time of the $i$th node in a history only depends on $i$, not on the
history. In the example above, a symmetric timing would only consist of $6$
times, one for each depth of the tree.
\begin{prop}\label{prop:sym}
  If a symmetric choiceless game has an $\epsilon$-timing using time at
  most $N$, then it has a \emph{symmetric} $\epsilon$-timing using time at
  most $N$.
\end{prop}

In order to prove the proposition we need to define a notation for timings of symmetric choiceless games. The point of the proposition is to show the we do not need to consider general timings of symmetric choiceless games, so the definition of timings will only be used for this one proof. Only the definition of symmetric choiceless games will be used later.

\begin{defi}\label{defi:temp}
A \emph{symmetric choiceless game} is given by a number $n$ and a finite sequence $\Gamma$ of length $|\Gamma|$ over the set of players, $\{1,\dots ,n\}$. Let $\Gamma_i$ denote the element in the $i$'th position in the sequence $\Gamma$.

A \emph{timing} of a symmetric choiceless game is a random tuple $X=(X_{\sigma,i})_{\sigma\in S_n,i\in [|\Gamma|]}$ of non-negative reals. For each $\sigma$ and $i$ we require that $X_{\sigma,i}+1\leq X_{\sigma,i+1}$. A \emph{symmetric timing} is a timing where $X_{\sigma,i}$ never depends on $\sigma$.

When using set notation we write a subscript $j$ after the set to mean the $j$'th smallest elements of the set.
An \emph{$\epsilon$-timing} of $\Gamma$ is a timing such that for any player, $p$, and any two assignment of numbers to the players, $\sigma,\sigma'\in S_n$ and any $j\leq |\{X_{\sigma,i}|\Gamma_i=\sigma(p)\}|,|\{X_{\sigma',i}|\Gamma_i=\sigma'(p)\}|$ we have $\{X_{\sigma,i}|\Gamma_i=\sigma(p)\}_j\stackrel{\epsilon}{\sim}\{X_{\sigma',i}| \Gamma_i=\sigma'(p)\}_j$.
\end{defi}

\begin{proof}
Fix a symmetric choiceless game $\Gamma$ on $n$ players. Given an $\epsilon$-timing $X$ we construct a symmetric $\epsilon$-timing $X'$ of the game: Let $\Pi$ be a random variable that is uniformly distributed on $S_n$ independently from $X$. Now define $x'(x,\pi)$ by  $x'(x,\pi)_{\sigma,i}=x_{\pi,i}$ and let $X'=x'(X,\Pi)$. This is clearly a symmetric timing. 

We want to show that it is an $\epsilon$-timing. To do this, fix a player, $p$, and two histories. A history is given by the assignment of numbers to players, so let $\sigma$ and $\sigma'$ be the assignments that results in those two histories. We also fix two number $j$ which indicate how far we go in the history. We assume $j\leq |\{i|\Gamma_i=\sigma(p)\}|,|\{i|\Gamma_i=\sigma'(p)\}|$ As $\Pi$ is uniformly distributed we have
\[ \{X_{\Pi,i}|\Gamma_i = \sigma(p) \}_j\sim \{X_{\sigma\circ \Pi,i  } | \Gamma_i=\sigma(\Pi(\Pi^{-1}(p))) \}_j \]
and similar for $\sigma'$. As $X$ is an $\epsilon$-timing we know that
\[\{X_{\sigma,i}|\Gamma_i=\sigma(p')\}_j\stackrel{\epsilon}{\sim}\{X_{\sigma',i}|\Gamma_i=\sigma'(p')\}_j\]
for all $p',\sigma $ and $\sigma'$ where these sets contain at least $j$ elements. By substituting $\sigma\circ \pi$ for $\sigma$, $\sigma'\circ \pi$ for $\sigma'$ and $\pi^{-1}(p)$ for $p$, where $\pi\in S_n$, we get 
\[\{X_{\sigma\circ\pi,i}|\Gamma_i=\sigma(\pi(\pi^{-1}(p)))\}_j\stackrel{\epsilon}{\sim}\{X_{\sigma'\circ \pi,i}|\Gamma_i=\sigma'(\pi(\pi^{-1}(p)))\}_j.\]
By Proposition \ref{prop:disjoint} this implies 
\[\{X_{\sigma\circ\Pi,i}|\Gamma_i=\sigma(\Pi(\Pi^{-1}(p)))\}_j\stackrel{\epsilon}{\sim}\{X_{\sigma'\circ \Pi,i}|\Gamma_i=\sigma'(\Pi(\Pi^{-1}(p)))\}_j.\]
Putting it all together gives
\begin{align*}
\{X'_{\sigma,i}|\Gamma_i=\sigma(p)\}_j=& \{X_{\Pi,i}|\Gamma_i= \sigma(p) \}_j\\
  \sim & \{X_{\sigma\circ \Pi,i  } | \Gamma_i=\sigma(\Pi(\Pi^{-1}(p))) \}_j\\
\stackrel{\epsilon}{\sim}&\{X_{\sigma'\circ \Pi,i}|\Gamma_i=\sigma'(\Pi(\Pi^{-1}(p)))\}_j\\
\sim&\{X_{\Pi,i}|\Gamma_i = \sigma'(p)\}_j\\
=& \{X'_{\sigma',i}|\Gamma_i=\sigma'(p)\}_j.
\end{align*}
\end{proof}

Hence, to show the lower bound, we only need to show that for any $r$ there
exists a symmetric choiceless game that does not have a \emph{symmetric}
$\epsilon$-timing using time
$\exp_2^r\left(O\left(\frac{1}{\epsilon}\right)\right)$. We will use
induction, but in order to make the induction hypothesis stronger, we will
show
 something stronger: even if we allow some freedom in which
order we schedule the nodes, we cannot find a good timing. To formalize
this, we use something we call agendas. An example of an agenda is
$2|33|11|2$, where the ``$|$'' are called separators. A timing of this
agenda would give a time to each number and each separator. The times of
e.g. the $3$'s do \emph{not} have to be between those of the first and second
separator, but they cannot be assigned times that are ``too much'' on the
wrong side of the separators. Furthermore, two separators will always have
times that are at least one apart.

\begin{defi}\label{defi:agenda}
  For players $\{1,\dots, n\}$, an \emph{agenda} $A$ is a finite sequence
  over $\{|,1,\dots, n\}$. Let $k_i$ be the number of times $i$ occurs in
  the list, and $k_|$ the number of times $|$ occurs in the list. An
  $(\epsilon,\lambda)$-timing of $A$ is a random variable
  $X=(X_{|,1},\dots, X_{|,k_{|}},X_{1,1},\dots,
  X_{1,k_1},X_{2,1},\dots,X_{2,k_2},\dots, X_{n,1},\dots, X_{n,k_n})$ which
  is a tuple of the same length as $A$, satisfying the following
  requirements.
\begin{enumerate}
\item Each $X_{i,j}$ takes real numbers and each $X_{|,j}$ takes non-negative values.
\item For all $j\in \{1,\dots, k_{|}-1\}$ we have $X_{|,j}+1\leq X_{|,j+1}$.
\item For all $i$ and $j\in\{1,\dots,k_i-1\}$ we have $X_{i,j}<X_{i,j+1}$.
\item If the $j_1$th occurrence of $i$ in $A$ is before the $j_2$th occurrence of $|$ then $X_{i,j_1}\leq X_{|,j_2}+\lambda$.
\item If the $j_1$th occurrence of $i$ in $A$ is after the $j_2$th occurrence of $|$ then $X_{i,j_1} \geq X_{|,j_2}-\lambda$.
\item If $k_i=k_j$ then $\delta(X_i,X_j)\leq \epsilon$.
\end{enumerate}
\end{defi}
Notice that the only requirement that compares $X_{i,j}$s is requirement
$3$, and it simply says that for fixed $i$ the $X_{i,j}$s should be
ordered. Otherwise, the $X_{i,j}$s are only ``kept in place'' by the
$X_{|,j}$s. Also note that requirement 6 implies that
$\delta((X_{i,1},\dots, X_{i,k}),(X_{j,1},\dots,X_{j,k}))\leq\epsilon$
whenever $k\leq k_i=k_j$, but unlike for games we have no such assumption
when $k_i\neq k_j$. For an agenda $A$, we can create a game $\Gamma_A$ by
removing the separators. Then a symmetric $\epsilon$-timing of $\Gamma_A$
will give an $(\epsilon,\lambda)$-timing of $A$ for any $\lambda\geq
0$. Hence, lower bounds on the time needed to time agendas gives lower
bounds on the time needed to time games.

The following proposition says that we can assume that all times $X_{i,j}$ are non-negative. However, when using the induction hypotheses, it will be useful that we allow $X_{i,j}$ to be negative. 

\begin{prop}\label{prop:nonneg}
Let $N,\epsilon$ and $\lambda<N$ be three numbers. 
If the agenda $A$ has an $(\epsilon,\lambda)$-timing with $X_{i,j},X_{|,j}\leq N$ for all $i$ and $j$, then $A$ has an $(\epsilon,\lambda)$-timing where $0\leq X_{i,j},X_{|,j}\leq N$ for all $i$ and $j$.
\end{prop}
\begin{proof}
Given a timing $X$ of $A$ we want to construct a timing $X'$ of $A$ that only contains non-negative times. Let $f:(-\infty,\lambda]\to (0,\lambda]$ be some increasing bijection with $f(x)\geq x$ and define $X'_{|,j}=X_{|,j}$ and $X'_{i,j}=f(X_{i,j})$ if $X_{i,j}\leq \lambda$ and $X'_{i,j}=X_{i,j}$ otherwise. As $X$ is a timing of $A$ we know that $X'_{|,j}=X_{|,j}\geq 0$ for all $j$, so $X'$ satisfies requirement $4$ from the definition of $(\epsilon,\lambda)$-timing. As each $X'_i$ can be obtained as a function of $X_i$ requirement $6$ is satisfied. It is easy to see that $X'$ satisfies the last $4$ requirements.
\end{proof}

The following lemma will be the base case of the induction proof. Here and
what follows, $\log$ denotes the base $2$ logarithm.
\begin{lemm}\label{lemm:agendabase}
For any $\left(\epsilon,\frac{1}{10}\right)$-timing of the agenda $2|33|11|2$ we have $\Pr\left(X_{2,2}\geq \exp_2\left((2\epsilon)^{-1}\right)\right)>0$. 
\end{lemm}
\begin{proof}
Let $\lambda=\frac{1}{10}$. By Proposition \ref{prop:nonneg} we can assume that $X_{i,j}\geq 0$ for all $i,j$.
From $X_{3,1}\geq X_{|,1}-\lambda$ and $X_{3,2}\leq X_{|,2}+\lambda$ and $X_{|,1}+1\leq X_{|,2}$ we get
\begin{align*}
\log(X_{3,2}-X_{3,1})\leq& \log(X_{|,2}+\lambda-(X_{|,1}-\lambda))\\
=&\log(X_{|,2}-X_{|,1})+\log\left(\frac{X_{|,2}-X_{|,1}+2\lambda}{X_{|,2}-X_{|,1}}\right)\\
=&\log(X_{|,2}-X_{|,1})+\log\left(1+\frac{2\lambda}{X_{|,2}-X_{|,1}}\right)\\
\leq&\log(X_{|,2}-X_{|,1})+\log\left(1+2\lambda\right)\\
\leq & \log(X_{|,2}-X_{|,1})+0.3.
\end{align*}
Similarly, $\log(X_{1,2}-X_{1,1})\leq \log(X_{|,3}-X_{|,2})+0.3$ and
\begin{align*}
\log(X_{2,2}-X_{2,1})\geq& \log(X_{|,3}-\lambda-(X_{|,1}+\lambda))\\
=&\log(X_{|,3}-X_{|,1})+\log\left(\frac{X_{|,3}-X_{|,1}-2\lambda}{X_{|,3}-X_{|,1}}\right)\\
=&\log(X_{|,3}-X_{|,1})+\log\left(1-\frac{2\lambda}{X_{|,3}-X_{|,1}}\right)\\
\geq &\log(X_{|,3}-X_{|,1})+\log\left(1-\lambda\right)\\
\geq &\log(X_{|,3}-X_{|,1})-0.2.
\end{align*}
We have $X_{|,3}-X_{|,1}=(X_{|,3}-X_{|,2})+(X_{|,2}-X_{|,1})$
so by Jensen's inequality we get
\begin{align*}
\log(X_{|,3}-X_{|,1})-1=&\log\left(\frac{X_{|,3}-X_{|,1}}{2}\right) 
\geq \frac{\log(X_{|,3}-X_{|,2})+\log(X_{|,2}-X_{|,1})}{2}.
\end{align*}

Thus,
%vc: in the pdf it looks like there is an operator called E2 -- maybe add parens?
%skj: I added \,
\begin{align*}
\E \,2\log(X_{2,2}-X_{2,1})
\geq& \E \,2\log(X_{|,3}-X_{|,1})-0.4\\
\geq& \E \log(X_{|,3}-X_{|,2})+\log(X_{|,2}-X_{|,1})+1.6\\
\geq& \E \log(X_{1,2}-X_{1,1})+\log(X_{3,2}-X_{3,1})+1.
\end{align*}
So we must have at least one of $\E \log(X_{2,2}-X_{2,1})\geq \E \log(X_{1,2}-X_{1,1})+0.5$ and $\E \log(X_{2,2}-X_{2,1})\geq \E \log(X_{3,2}-X_{3,1})+0.5$. Assume without loss of generality that the former is the case. As $X$ is an  $\epsilon$-timing we have $\delta((X_{2,1},X_{2,2}),(X_{1,1},X_{1,2}))\leq 
\epsilon$ so by Proposition \ref{prop:dpi} and \ref{prop:tooln2} the $\log$s must take values in an interval of length at least $\frac{1}{2\epsilon}$. As the $X_{i,j}$s always differ by at least one, the $\log$s only take nonnegative values. Hence, $\log(X_{2,2}-X_{2,1})\geq \frac{1}{2\epsilon}$ with positive probability, so $X_{2,2}\geq X_{2,2}-X_{2,1}\geq 2^{\frac{1}{2\epsilon}}$ with positive probability. 
\end{proof} 
 
 \begin{lemm}\label{lemm:gapsbetweensep}
 Let $\epsilon,\lambda>0$ and $c>2$ be given. If $X$ is an $(\epsilon,\lambda)$-timing of the agenda $24|33|1441|22|13$, there is probability at least $1-(8c+\frac{3}{1-\frac{2}{c}})(2\lambda+\epsilon)$ that either
 \begin{enumerate}
 \item$X_{|,3}< \frac{(c-1)X_{|,1}+X_{|,4}}{c}$ and $X_{|,2}< \frac{(c-1)X_{|,1}+X_{|,3}}{c}$, or
\item $X_{|,2}> \frac{X_{|,1}+(c-1)X_{|,4}}{c}$ and $X_{|,3}> \frac{X_{|,2}+(c-1)X_{|,4}}{c}$.
\end{enumerate}
 \end{lemm}
 \begin{proof}
 Define $f(x,y,z)=\frac{y-x}{z-x}$. Let $X_{|,-i}$ denote the tuple $X_|$ with $X_{|,i}$ removed. Then possibility $1$ from the statement of the lemma is equivalent to $f(X_{|,-2}), f(X_{|,-4})<\frac{1}{c}$ and possibility $2$ is equivalent to $f(X_{|,-3}),f(X_{|,-1})>1-\frac{1}{c}$. 
We have
 \begin{align*}
 f(X_{|,-4})f(X_{|,-2})=&\frac{X_{|,2}-X_{|,1}}{X_{|,3}-X_{|,1}}\frac{X_{|,3}-X_{|,1}}{X_{|,4}-X_{|,1}}
 =\frac{X_{|,2}-X_{|,1}}{X_{|,4}-X_{|,1}}
 =f(X_{|,-3}).
 \end{align*}

 We see that for $x<y<z$ the function $f$ is increasing in $y$ and decreasing in $x$ and $z$ and it gives values in $(0,1)$. Furthermore, $f(x+a,y+a,z+a)=f(x,y,z)$. We now have 
 \begin{align*}
 f(X_{|,-4})=&f(X_{|,1}+\lambda,X_{|,2}+\lambda,X_{|,3}+\lambda)\\
 =&f(X_{|,1}+\lambda,X_{|,2}-\lambda,X_{|,3}+\lambda)+\frac{2\lambda}{X_{|,3}-X_{|,1}}\\
 \leq& f(X_{4,1},X_{4,2},X_{4,3})+\lambda\\
 =&f(X_4)+\lambda.
 \end{align*}
 Here the inequality follows from the inequalities between the $X_{|,i}$s and the $X_{4,i}$s, and from $X_{|,3}-X_{|,1}\geq 2$. Similarly, we get $f(X_{|,-2})\leq f(X_2)+\lambda$ and $f(X_{|,-3})\geq f(X_3)-\lambda$. 
 
 As $X$ is an $(\epsilon,\lambda)$-timing and $f$ takes values in $[0,1]$, we get $\E \left(f(X_4)-f(X_3)\right)\leq \epsilon$. So
 \begin{align*}
 \E\left( f(X_{|,-4})-f(X_{|,-3})\right)=&\E \left(f(X_{|,-4})-f(X_4)+f(X_4)-f(X_3)+f(X_3)-f(X_{|,-3})\right)\\
 \leq& \lambda + \epsilon+\lambda\\ %space
 =& %space
 2\lambda +\epsilon,
 \end{align*}
 and similarly $\E\left( f(X_{|,-2})-f(X_{|,-3})\right)\leq 2\lambda+\epsilon$. In the other direction, the monotonicity of $f$ implies that $f(X_{|,-4})>f(X_{|,-3})$ and $f(X_{|,-2})>f(X_{|,-3})$. 
 
 Using
 $f(X_{|,-4})f(X_{|,-2})=f(X_{|,-3})$ we get 
 \begin{align*}
 \E &\left( f(X_{|,-4})(1-f(X_{|,-2}))+f(X_{|,-2})(1-f(X_{|,-4}))\right)\\
 &=\E \left(f(X_{|,-4})+f(X_{|,-2})-2f(X_{|,-2})f(X_{|,-4})\right)\\
 &=\E \left(f(X_{|,-4})-f(X_{|,-3})+f(X_{|,-2})-f(X_{|,-3})\right)\\
 &\leq 4\lambda+2\epsilon.
 \end{align*}
 If $f(x_{|,-4})\in\left[\frac{1}{c},1-\frac{1}{c}\right]$ we see that 
 \[f(x_{|,-4})(1-f(x_{|,-2}))+f(x_{|,-2})(1-f(x_{|,-4}))\geq \frac{1}{c}(1-f(x_{|,-2}))+f(x_{|,-2})\frac{1}{c}=\frac{1}{c}.\]
 By Markov's inequality $\Pr(f(X_{|,-4})\in \left[\frac{1}{c},1-\frac{1}{c}\right])\leq (4\lambda+2\epsilon)c$. Similarly for $X_{|,-2}$. 
 
 The agenda $24|33|1441|22|13$ has a symmetry: if we swap $1$ and $4$ and swap $2$ and $3$ we reverse the order of the numbers (ignoring the order of numbers between two separators). Hence, we can go through the above proof with $1-f(X_{|,-1})$ instead of $f(X_{|,-4})$, and $1-f(X_{|,-2})$ instead of $f(X_{|,-3})$ and so on. This would give us $\Pr(f(X_{|,-1})\in \left[\frac{1}{c},1-\frac{1}{c}\right])\leq (4\lambda+2\epsilon)c$ and $\Pr(f(X_{|,-3})\in \left[\frac{1}{c},1-\frac{1}{c}\right])\leq (4\lambda+2\epsilon)c$. 
 
 We know that $f(X_{|,-4})>f(X_{|,-3})$ and $\E\left(
   f(X_{|,-4})-f(X_{|,-3})\right)\leq 2\lambda+\epsilon$, so by Markov's
 inequality, the probability that $f(X_{|,-4})$ and $f(X_{|,-3})$ are both
 in $[0,\frac{1}{c})\cup (1-\frac{1}{c},1]$ but in two different parts of
 it is at most $\frac{2\lambda+\epsilon}{1-\frac{2}{c}}$. Similarly for the
 pair $f(X_{|,-3})$ and $f(X_{|,-2})$ and for the pair $f(X_{|,-2})$ and
 $f(X_{|,-1})$. Let $T$ be the random variable that is $1$ if for all $i\in
 [4]$ we have $f(X_{|,-i})\in \left[0,\frac{1}{c}\right)$, or if for all
 $i\in[4]$ we have $f(X_{|,-i})\in \left(1-\frac{1}{c},1\right]$. Otherwise
 $T=0$. By the union bound, $\Pr(T=0)\leq
 4(4\lambda+2\epsilon)c+3\frac{2\lambda+\epsilon}{1-\frac{2}{c}}=(8c+\frac{3}{1-\frac{2}{c}})(2\lambda+\epsilon)$. When
 $T=1$ we either have $f(X_{|,-i})$ for all $i$, and the first case of the
 conclusion holds, or $f(X_{|,-i})>1-\frac{1}{c}$, and the second case of
 the conclusion holds.
 \end{proof}

\begin{prop}\label{prop:allthesame2}
Let $c>2$ and $x_1<\dots <x_n$ be a sequence such that for all $i\in [n-3]$ we have either
\begin{enumerate}
\item $x_{i+2}<\frac{(c-1)x_i+x_{i+3}}{c}$ and $x_{i+1}< \frac{(c-1)x_i+x_{i+2}}{c}$, or
\item $x_{i+1}> \frac{x_i+(c-1)x_{i+3}}{c}$ and $x_{i+2}> \frac{x_{i+1}+(c-1)x_{i+3}}{c}$.
\end{enumerate}
Then it must be the same of the two conditions that holds for every $i$. If it is the first then $x_{i+1}-x_i\geq \left(c-\frac{1}{c}-1\right)(x_i-x_1)$ for all $i\in\{2,\dots, n-1\}$. If it is the second then $x_i-x_{i-1}\geq \left(c-\frac{1}{c}-1\right)(x_n-x_i)$ for all $i\in\{2,\dots n-1\}$.
\end{prop}
  \begin{proof}
If the first equality holds for some $i$ we have $x_{i+2}<\frac{(c-1)x_i+x_{i+3}}{c}<\frac{(c-1)x_{i+1}+x_{i+3}}{c}<\frac{x_{i+1}+x_{i+3}}{2}$ and $x_{i+1}\leq \frac{(c-1)x_i+x_{i+2}}{c}<\frac{x_i+x_{i+2}}{2}$. This implies $x_{i+1}-x_i<x_{i+2}-x_{i+1}<x_{i+3}-x_{i+2}$. Similarly, if the second set of inequalities holds for some $i$, the gaps get smaller. Hence, if $x_{i+1}-x_i<x_{i+2}-x_{i+1}$ we know that for both $i-1$ and $i$ we must be in the first case and if $x_{i+1}-x_i>x_{i+2}-x_{i+1}$ then for both $i-1$ and $i$ we must be in the second case, and we cannot have $x_{i+1}-x_i=x_{i+2}-x_{i+1}$. So for any two consecutive $i$s we must be in the same case, and by induction it must be the same case for all $i$.

To prove the second part of the proposition, assume that the first case holds for all $i$. We want to show that $x_{i+1}-x_i\geq \left(c-\frac{1}{c}-1\right)(x_i-x_1)$ for all $i\in\{2,\dots, n-1\}$ by induction on $i$. By substituting $i=1$ into the second inequality of the first case, we get $x_{2}\leq \frac{(c-1)x_1+x_3}{c}$, so $\left(c-\frac{1}{c}-1\right)(x_2-x_1)\leq(c-1)(x_2-x_1)\leq x_3-x_2$. By substituting $i=1$ into the first inequality of the case, we obtain $cx_{3}<(c-1)x_1+x_4$ which implies $x_4-x_3>(c-1)(x_3-x_1)>\left(c-1-\frac{1}{c}\right)(x_3-x_1)$. Thus, the induction hypothesis is true for $i=2$ and $i=3$.

Next, we assume for induction that $x_{i}-x_{i-1}\geq \left(c-\frac{1}{c}-1\right)(x_{i-1}-x_1)$ and we want to prove $x_{i+2}-x_{i+1}\geq \left(c-\frac{1}{c}-1\right)(x_{i+1}-x_1)$, so we use the induction hypothesis for $i-1$ to prove it for $i+1$. By substituting $i-1$ instead of $i$ in the first inequality of the first case we get $cx_{i+1}< (c-1)x_{i-1}+x_{i+2}$, so 
\begin{align*}
x_{i+2}-x_{i+1}> &(c-1)(x_{i+1}-x_{i-1})\\
=& \left(c-1-\frac{1}{c}\right)(x_{i+1}-x_{i-1})+\frac{1}{c}(x_{i+1}-x_{i-1})\\
\geq &  \left(c-1-\frac{1}{c}\right)(x_{i+1}-x_{i-1}) +(x_i-x_{i-1})\\
\geq &  \left(c-1-\frac{1}{c}\right)(x_{i+1}-x_{i-1}) +\left(c-1-\frac{1}{c}\right)(x_{i-1}-x_{1}) \\
\geq &   \left(c-1-\frac{1}{c}\right)(x_{i+1}-x_{1}).
\end{align*}
Here the second inequality follows from $x_{i}< \frac{(c-1)x_{i-1}+x_{i+1}}{c}$ and the third inequality follows from the induction hypothesis. 

The proof that $x_i-x_{i-1}\geq \left(c-\frac{1}{c}-1\right)(x_n-x_i)$ if the second case holds for all $i$ is similar. 
\end{proof}

\begin{theo}\label{theo:lower}
Given $r\geq 1$ there exist $k_r,\epsilon_r>0$ and an agenda $A_r$ with $16r-13$ players and at most $3(r-1)$ nodes per player such that for any $\left(\epsilon,\frac{\epsilon^{1-\frac{1}{r}}}{10}\right)$-timing of $A_r$ with $\epsilon\leq\epsilon_r$ we need time at least $\exp_2^r\left(k_r\epsilon^{-\frac{1}{r!(r-1)!}}\right)$. 
\end{theo}
\begin{proof}
We construct the games recursively. We define $A_1=2|3332|111|2$.
To construct $A_{r}$ for $r>1$ we first take $A_{r-1}$, and add two separators on either side. E.g., for $r=2$ this gives us  
$||2|3332|111|2||.$
For each of the players in $A_{r-1}$, we then give them a node between the first two separators and a node between the last two. In the example, this gives
$|123|2|3332|111|2|123|.$
Around the first $4$ separators we put four new players $\text{a},\text{b},\text{c},\text{d}$ to ensure that the gaps between consecutive separators either increase or decrease using the pattern $\text{bd}|\text{cc}|\text{adda}|\text{bb}|\text{ac}$. For $r=2$ we now have 
\[\text{bd}|\text{cc}123|\text{adda}2|\text{bb}3332|\text{ac}111|2|123|.\]
Then around separators $2$ to $5$ we then put the player $\text{e},\text{f},\text{g},\text{h}$. The example becomes 
\[\text{bd}|\text{ccfh}123|\text{addagg}2|\text{bbehhe}3332|\text{acff}111|\text{eg}2|123|.\]
Similarly, we put the players $\text{i},\text{j},\text{k},\text{l}$ around separators $3$ to $6$ and the players $\text{m},\text{n},\text{o},\text{p}$ around separators $4$ to $7$. This gives
\[\text{bd}|\text{ccfh}123|\text{addaggjl}2|\text{bbehhekknp}3332|\text{acffillioo}111|\text{egjjmppm}2|\text{iknn}123|\text{mo}\]
If there are more groups of $4$ consecutive separators left, we could continue the same way. However, we can now reuse players, so player $\text{a},\text{b},\text{c},\text{d}$ will be used around separators $5$ to $8$ and so on. This uses $16$ players more than for $A_{r-1}$, so by induction $A_r$ has $16r-13$ players. The new players each have $3\left\lfloor\frac{4r-1}{4}\right\rfloor=3(r-1)$ nodes, and all other players have $2$ nodes more than in the $r-1$ game. By induction each player has at most $3(r-1)$ nodes.

Now we want to show that we need a long time to time these agendas. 
We know from Lemma \ref{lemm:agendabase} that any
$\left(\epsilon,\frac{1}{10}\right)$-timing of $2|33|11|2$ will need time
at least $\exp_2((2\epsilon)^{-1})$. Furthermore, from any
$\left(\epsilon,\frac{1}{10}\right)$-timing of $2|3332|111|2$ we can throw
away the middle time for each player and 
get an $\left(\epsilon,\frac{1}{10}\right)$-timing of $2|33|11|2$. This proves the base case. (The middle node is needed in the induction step, to ensure that player $2$ can see if the gaps between separators increase or decrease).

Next let $r>1$ and assume that the statement holds for $r-1$. Define 
$\epsilon_r=\frac{\epsilon_{r-1}^{r(r-1)}}{180r}.$

Assume that we have an $(\epsilon,\frac{\epsilon^{1-\frac{1}{r}}}{10})$-timing $X$ of $A_r$ where $\epsilon<\epsilon_r$. By Proposition \ref{prop:nonneg} we can assume that $X_{i,j}\geq 0$ for all $i,j$. For each $i\in [4r-4]$ we claim that there is probability at least $1-\left(8\epsilon^{\frac{1}{r-1}-1}+\frac{3}{1-\frac{2}{\epsilon^{\frac{1}{r-1}-1}}}\right)(\frac{2\epsilon^{1-\frac{1}{r}}}{10}+\epsilon)$ that we have either 
 \begin{enumerate}
 \item$X_{|,i+2}< \frac{(\epsilon^{\frac{1}{r-1}-1}-1)X_{|,i}+X_{|,i+3}}{\epsilon^{\frac{1}{r-1}-1}}$ and $X_{|,i+1}< \frac{(\epsilon^{\frac{1}{r-1}-1}-1)X_{|,i}+X_{|,i+2}}{\epsilon^{\frac{1}{r-1}-1}}$, or
\item $X_{|,i+1}> \frac{X_{|,i}+(\epsilon^{\frac{1}{r-1}-1}-1)X_{|,i+3}}{\epsilon^{\frac{1}{r-1}-1}}$ and $X_{|,i+2}> \frac{X_{|,i+1}+(\epsilon^{\frac{1}{r-1}-1}-1)X_{|,i+3}}{\epsilon^{\frac{1}{r-1}-1}}$.
\end{enumerate}
If $i$ is of 
the form $4j+1$, we show this using the players
$\text{a},\text{b},\text{c},\text{d}$: As we have an
$(\epsilon,\frac{\epsilon^{1-\frac{1}{r}}}{10})$ timing of our agenda, we
can take the $(3j+1)$th, $(3j+2)$th, and $(3j+3)$th nodes of these four players
and the $(4j+1)$th, $(4j+2)$th, $(4j+3)$th and $(4j+4)$th separator and get an
$(\epsilon,\frac{\epsilon^{1-\frac{1}{r}}}{10})$-timing of the agenda
$\text{bd}|\text{cc}|\text{adda}|\text{bb}|\text{ac}$. The claim now follows from Lemma
\ref{lemm:gapsbetweensep}. If $i$ is of 
the form $4j+2$ we would use player $\text{e},\text{f},\text{g},\text{h}$ instead, and so on. This proves the claim.

Let $T$ be the random variable that is $1$ if one of the two cases holds
for each $i\in [4r-4]$ and $0$ otherwise. 
By the union bound
\begin{align*}
\Pr(T=0)\leq &(4r-4)\left(8\epsilon^{\frac{1}{r-1}-1}+\frac{3}{1-\frac{2}{\epsilon^{\frac{1}{r-1}-1}}}\right)(\frac{2\epsilon^{1-\frac{1}{r}}}{10}+\epsilon)\\
\leq& 4r\left(9\epsilon^{\frac{1}{r-1}-1}\right)1.2\epsilon^{1-\frac{1}{r}}\\
=&44r\epsilon^{\frac{1}{r-1}-\frac{1}{r}}.
\end{align*}
Let $X'=X|_{T=1}$ be the timing of the agenda given $T=1$. By Proposition \ref{prop:givenT} $X'$ is an $(\epsilon',\epsilon^{1-\frac{1}{r}})$-timing with
\[\epsilon'\leq \frac{\epsilon+44r\epsilon^{\frac{1}{r-1}-\frac{1}{r}}}{1-44r\epsilon^{\frac{1}{r-1}-\frac{1}{r}}}\leq 2\left(45r\epsilon^{\frac{1}{r-1}-\frac{1}{r}}\right)=90r\epsilon^{\frac{1}{(r-1)r}}.\]

Given $X'$ we will then construct a $\left(180r\epsilon^{\frac{1}{(r-1)r}},\left(180r\epsilon^{\frac{1}{(r-1)r}}\right)^{1-\frac{1}{r}}\right)$-timing of $A_{r-1}$. First we want to argue that for a timing $X'$, each of the old players (players from $A_{r-1}$) can figure out if the gaps between $X_{|,i}$s are increasing or decreasing. By induction, each player has  $k_i\geq 2r+1\geq 5$ nodes. Furthermore, in $A_{r-1}$ each player occurs at most once before the first separator, so in $A_r$ each old player occurs at most twice before the third separator, so $X_{i,3}\geq X_{|,3}-\frac{\epsilon^{1-\frac{1}{r}}}{10}$ (here and in the following $i$ denotes the number of an old player). Conversely, $X_{i,3}\leq X_{i,k_i-2}\leq X_{|,k_|-2}+\frac{\epsilon^{1-\frac{1}{r}}}{10}$. As each of the old players has a node between the first two separators and a node between the last two separators, we get $X_{i,1}\leq X_{|,2}+\frac{\epsilon^{1-\frac{1}{r}}}{10}$ and $X_{i,k_i}\geq X_{|,k_|-1}-\frac{\epsilon^{1-\frac{1}{r}}}{10}$. 
Let $c=\epsilon^{\frac{1}{r-1}-1}$. If we are in the case where gaps between separators increases, we get
\begin{align*}
X_{i,k_i}-X_{i,3}\geq& \left(X_{|,k_|-1}-\frac{\epsilon^{1-\frac{1}{r}}}{10}\right)-\left(X_{|,k_|-2}+\frac{\epsilon^{1-\frac{1}{r}}}{10}\right)\\
=&\left(X_{|,k_|-1}-X_{|,k_|-2}\right)-\frac{1}{5}\epsilon^{1-\frac{1}{r}}\\
\geq& \left(c-\frac{1}{c}-1\right)(X_{|,k_|-2}-X_{|,2})-\epsilon^{1-\frac{1}{r}}\\
\geq& \left(c-2\right)\left(X_{i,3}-\frac{\epsilon^{1-\frac{1}{r}}}{10}-X_{i,1}-\frac{\epsilon^{1-\frac{1}{r}}}{10}\right)-\epsilon^{1-\frac{1}{r}}\\
\geq&  \left(c-2\right)(X_{i,3}-X_{i,1})-c\epsilon^{1-\frac{1}{r}}\\
\geq& 3(X_{i,3}-X_{i,1})-\epsilon^{\frac{1}{r-1}-1}\epsilon^{1-\frac{1}{r}}\\
=& 3(X_{i,3}-X_{i,1})-\epsilon^{\frac{1}{r-1}-\frac{1}{r}}\\
\geq& 3(X_{i,3}-X_{i,1})-1\\
> & X_{i,3}-X_{i,1}.
\end{align*}
The second inequality follows from Proposition \ref{prop:allthesame2} on the sequence $x_{|,2},x_{|,3},\dots, x_{|,k_|}$. In the fifth inequality we use $c-2\geq 3$ which is equivalent to $\epsilon^{\frac{1}{r-1}-1}\geq 5$ and follows from $\epsilon\leq\epsilon_r=\frac{\epsilon_{r-1}^{r(r-1)}}{180r}<\frac{1}{25}$. 
Finally, we use $\epsilon<1$ and $X_{i,3}-X_{i,1}\geq X_{|,3}-X_{|,2}-\frac{2\epsilon^{1-\frac{1}{r}}}{10}\geq \frac{1}{2}$.

Conversely, if we are in the case where gaps between separators decrease, a similar computation shows that $X_{i,3}-X_{i,1}>  X_{i,k_i}-X_{i,3}$.
Thus, if one knows that $i$ is a player from $A_{r-1}$ and one knows
$X'_i$, one can then decide if the gaps between the $X'_{|,j}$s are
increasing or decreasing, even if one does not know $i$. Let $S$ be the
random variable that is $1$ if the gaps are increasing and $2$ if they are
decreasing. We must have either $\Pr(S=1)\geq \frac{1}{2}$ or $\Pr(S=2)\geq
\frac{1}{2}$. Assume without loss of generality that we have $\Pr(S=1)\geq
\frac{1}{2}$. As the supports 
of the $X_i|_{S=1}$s are disjoint from the supports of the $X_{j}|_{S=2}$s Proposition \ref{prop:disjoint} gives us 
\begin{align*}
\delta(X'_i,X'_j)=&\Pr(S=1)\delta(X'_i|_{S=1},X'_j|_{S=1})+\Pr(S=2)\delta(X'_i|_{S=2},X'_j|_{S=2})\\
\geq &\frac{1}{2}\delta(X'_i|_{S=1},X'_j|_{S=1}).
\end{align*} 
We define $X''=X'|_{S=1}$. By the above inequality we have $ \delta(X''_i,X''_j)\leq \delta(X'_i,X'_j)\leq 2\cdot 90r\epsilon^{\frac{1}{(r-1)r}}$ for all old players $i$ and $j$ with $k_i=k_j$.

We now define a timing $Y$ of $A_{r-1}$ by
\begin{itemize}
\item $Y_{i,j}=\log(X_{i,j+1}-X_{i,1})$ for all numbers
 $i,j$ such that player $i$ occurs at least $j$ times in $A_{r-1}$.
\item $Y_{|,j}=\log(X_{|,j+2}-X_{|,2})$ for all numbers
$j$ less than the number of separators in $A_{r-1}$.  
\end{itemize}
We want to show that this is a $\left(180r\epsilon^{\frac{1}{(r-1)r}},\frac{\left(180r\epsilon^{\frac{1}{(r-1)r}}\right)^{1-\frac{1}{r}}}{10}\right)$-timing of $A_{r-1}$. If we can show that, the induction hypothesis then implies that the highest value that the $Y_{p,i}$s take must be at least 
\begin{align*}\exp_2^{r-1}\left(k_{r-1}\left(436r\epsilon^{\frac{1}{(r-1)r}}\right)^{-\frac{1}{(r-1)!(r-2)!}}\right)
%=&\exp_2^{r-1}\left(k_{r}\left(\epsilon^{\frac{1}{(r-1)r}}\right)^{-\frac{1}{(r-1)!(r-2)!}}\right)\\
=&\exp_2^{r-1}\left(k_{r}\epsilon^{-\frac{1}{(r)!(r-1)!}}\right).
\end{align*}
As the $X_{i,1}$s are non-negative this implies that $X''$ and $X$ have to take values of at least 
$2^{\exp_2^{r-1}\left(k_{r}\epsilon^{-\frac{1}{(r)!(r-1)!}}\right)}=\exp_2^{r}\left(k_{r}\epsilon^{-\frac{1}{(r)!(r-1)!}}\right).$

To finish the proof we only need to check that $Y$ satisfies the $6$ requirements for an $(\epsilon'',\frac{\epsilon''^{1-\frac{1}{r-1}}}{10})$-timing of $A_{r-1}$ with $\epsilon''=180r\epsilon^{\frac{1}{(r-1)r}}$. 

First, for fixed $i$ the $X_{i,j}$s are increasing in $j$ so the $Y_{i,j}$s are real numbers. We also have $X_{|,j}+1\leq X_{|,j+1}$, so the $Y_{|,j}$s are non-negative. 

The second requirement says $Y_{|,j}+1\leq Y_{|,j+1}$. By using Proposition \ref{prop:allthesame2} on the sequence $x_{|,2},\dots, x_{|,k_|}$ we get 
\begin{align*}
X_{|,j+3}-X_{|,2}=&X_{|,j+3}-X_{|,j+2}+X_{|,j+2}-X_{|,2}\\
\geq& \left(c-\frac{1}{c}-1\right)(X_{|,j+2}-X_{|,2})+X_{|,j+2}-X_{|,2}\\
=&\left(c-\frac{1}{c}\right)(X_{|,j+2}-X_{|,2})\\
\geq& 2(X_{|,j+2}-X_{|,2})
\end{align*}
where $c=\epsilon^{\frac{1}{r}-1}\geq 3$. Hence, 
\begin{align*}
Y_{|,j+1}=&\log(X_{|,j+3}-X_{|,2})
\geq \log(2(X_{|,j+2}-X_{|,2}))
=1+\log(X_{|,j+2}-X_{|,2})
=Y_{|,j}+1.
\end{align*} 

The third requirement says that for fixed $i$ the $Y_{i,j}$ should be increasing in $j$. This follows from the similar fact for  $X_{i,j}$.

To show the fourth requirement assume that the $j_1$'th occurrence of $i$ in $A_{r-1}$ is before the $j_2$'th occurrence of $|$. We need to show that $Y_{i,j_1}\leq Y_{|,j_2}+\frac{\epsilon''^{1-\frac{1}{r-1}}}{10}$. From the assumptions we have that the $j_1+1$th occurrence of $i$ in $A_r$ is before the $j_2+2$'th occurrence of $|$. As $X$ is an $\left(\epsilon,\frac{\epsilon^{1-\frac{1}{r-1}}}{10}\right)$-timing this implies that 
$X_{i,j_1+1}\leq X_{|,j_2+2}+\frac{\epsilon^{1-\frac{1}{r}}}{10}$,
and in general we have 
$X_{i,1}\geq X_{|,1}-\frac{\epsilon^{1-\frac{1}{r}}}{10}.$
Finally, because we are in the case where the gaps are increasing we have
$X_{|,2}-X_{|,1}\leq \frac{X_{|,3}-X_{|,2}}{c-1}\leq \frac{X_{|,j_2+2}-X_{|,2}}{c-1}$
so 
\begin{align*}
X_{|,j_2+2}-X_{|,1}=&X_{|,j_2+2}-X_{|,2}+X_{|,2}-X_{|,1}\\
\leq& \frac{c}{c-1}(X_{|,j_2+2}-X_{|,2}).
\end{align*}
Thus, we have 
\begin{align*}
Y_{i,j_1}=&\log(X_{i,j_1+1}-X_{i,1})\\
\leq & \log\left(X_{|,j_2+2}+\frac{\epsilon^{1-\frac{1}{r}}}{10}-X_{|,1}+\frac{\epsilon^{1-\frac{1}{r}}}{10}\right)\\
=&\log(X_{|,j_2+2}-X_{|,1})+\log\left(\frac{X_{|,j_2+2}-X_{|,1}+\frac{1}{5}\epsilon^{1-\frac{1}{r}}}{X_{|,j_2+2}-X_{|,1}}\right)\\
\leq & \log(X_{|,j_2+2}-X_{|,2})+\log\left(\frac{c}{c-1}\right)+\log\left(1+\frac{\epsilon^{1-\frac{1}{r}}}{10}\right)\\
\leq & Y_{|,j_2}+\log(1+\frac{2}{c})+\frac{2}{10}\epsilon^{1-\frac{1}{r}}\\
\leq & Y_{|,j_2}+4\epsilon^{1-\frac{1}{r-1}}+\frac{1}{5}\epsilon^{1-\frac{1}{r}}\\
\leq & Y_{|,j_2}+5\epsilon^{1-\frac{1}{r-1}}\\
\leq & Y_{|,j_2}+\frac{(\epsilon'')^{1-\frac{1}{r-1}}}{10}.
\end{align*}
The last inequality follows from $\epsilon''=180r\epsilon^{\frac{1}{(r-1)r}}$. 

To show the fifth requirement assume that $j_1$'th occurrence of $i$ in $A_{r-1}$ is after the $j_2$'th occurrence of $|$. We need to show that $Y_{i,j_1}\geq Y_{|,j_2}-\frac{\epsilon''^{1-\frac{1}{r-1}}}{10}$. The assumption implies that the $j_1+1$'th occurrence of $i$ in $A_r$ is after the $j_2+2$'th occurrence of $|$, so we have
\begin{align*}
Y_{i,j_1}=& \log(X_{i,j_1+1}-X_{i,1})\\
\geq & \log(X_{|,j_2+2}-\frac{\epsilon^{1-\frac{1}{r}}}{10}-X_{|,2}-\frac{\epsilon^{1-\frac{1}{r}}}{10})\\
\geq &\log(X_{|,j_2+2}-X_{|,2})+\log\left(\frac{X_{|,j_2+2}-X_{|,2}-\frac{1}{5}\epsilon^{1-\frac{1}{r}}}{X_{|,j_2+2}-X_{|,2}}\right)\\
\geq & Y_{|,j}+\log\left(1-\frac{1}{10}\epsilon^{1-\frac{1}{r}}\right)\\
\geq & Y_{|,j}-\frac{2}{10}\epsilon^{1-\frac{1}{r}}\\
\geq & Y_{|,j}-\frac{\epsilon''^{1-\frac{1}{r}}}{10}.
\end{align*}
Here the second to last inequality follows from $\frac{\epsilon^{1-\frac{1}{r}}}{10}\leq \frac{1}{2}$.

Finally, the last requirement says that if $k_i=k_j$ then $\delta(Y_i,Y_j)\leq \epsilon''$. This follows from the fact that $\delta(X_i'',X_j'')\leq \epsilon''$ and the fact that $Y_i$ is a function of $X_i''$. 
\end{proof}

\begin{coro}\label{coro:lowerbound}
Given $r\geq 1$ there exists $\epsilon_r>0$ and a game with $16r+3$ players and at most $3r$ nodes per player per history such that for any $\epsilon$-timing of the game with $\epsilon\leq\epsilon_r$ we need time at least $\exp_2^r\left(\epsilon^{-1}\right)$. 
\end{coro}
\begin{proof}
Let $r$ be given. We know from Theorem \ref{theo:lower} that there exist
$k_{r+1},\epsilon'_{r+1}>0$ and an agenda $A_{r+1}$ with $16r+3$ players
and at most $3r$ nodes per player such that for any
$(\epsilon,\epsilon^{r+1})$-timing of $A_{r+1}$ with
$\epsilon\leq\epsilon'_{r+1}$ we need time at least
$\exp_2^{r+1}\left(k_{r+1}\epsilon^{-\frac{1}{(r+1)!r!}}\right)$. We know
that for all $c>0$ we have $\log(x)\leq x^c$ for all sufficiently large
$x$. Using this for $x=\epsilon^{-1}$ and $c=\frac{1}{(r+1)!r!}$ we get
$\log(\epsilon^{-1})\leq k_{r+1}\epsilon^{-\frac{1}{(r+1)!r!}}$ for
sufficiently small $\epsilon$. By exponentiating on both sides we get
$\exp\left(k_{r+1}\epsilon^{-\frac{1}{(r+1)!r!}}\right)\geq
\epsilon^{-1}$. Thus, there is some $\epsilon_r$ such that for any
$(\epsilon,\epsilon^{1-\frac{1}{r}})$-timing of $A_{r+1}$ with
$\epsilon\leq\epsilon_r$ we need time at least
$\exp_2^r\left(\epsilon^{-1}\right)$. Let $\Gamma_r$ be the symmetric
choiceless game obtained from $A_{r+1}$ by removing the separators. Now any $\epsilon$-timing of $\Gamma_r$ must be an $(\epsilon,0)$-timing, and hence also an $(\epsilon,\epsilon^{1-\frac{1}{r}})$-timing of $A_{r+1}$. Thus, any $\epsilon$-timing of $\Gamma_r$ must use time at least $\exp_2^r\left(\epsilon^{-1}\right)$.
\end{proof}

%I think it is possible to show that such a game only need to have $2$ possible histories, but I haven't proved it. 
%Beginning of attempted proof that two histories are enough:
%The game $\Gamma_r$ has $(r+1)!$ different histories. To get a game with only $2$ histories, we define an equivalence relation on $[r+1]$: two numbers $i$ and $j$ are equivalent if $k_i=k_j$. For each equivalence class $\{i|k_i=k\}$ we define a cycle $\sigma_k$ on this class, and define $\sigma=\prod_k \sigma_k$. We now define a choiceless game $\Gamma'_r$ where Chance first chose uniformly at random between the identity assignment of numbers and $\sigma$. 

\section{Imperfect timekeeping}\label{sec:imperfect}
Previously we assumed that at any time all the players knew the exact
time. In practice, this is not a realistic assumption. Even our model of
time---that there exists an absolute time, and that everybody's time goes
at the same speed---has been proven wrong by relativity theory. If the
players cannot feel acceleration, one could use the twin paradox to time
games that otherwise cannot be exactly timed~\cite{twinbook}.\footnote{The
  question whether it is possible to implement a not exactly timeable game
  on players who are equipped with a perfect accelerometer is beyond the
  scope of this paper.} A more down-to-earth objection is that it might be
possible affect humans' or even computers' perception of time if you
control their environment.  The purpose of this section it to show that our
lower bounds are quite robust: even if we can determine the players'
perception of time within some reasonable bounds, there are games that take
a long time to $\epsilon$-time. We will assume each node occurs at some
``official'' time, $x$, 
%vc: according to which the choices are made?
%skj: I changed the sentence
 and that we can also decide the players' perception $y$ of that time. The following definition also models a situation where the players do not know when the game started.

\begin{defi}\label{defi:imperfecttime}
Let $l,u:\mathbb{R}^+\to \mathbb{R}^+$ be weakly increasing functions satisfying $l(t)\leq t\leq u(t)$. A deterministic $[l,u]$-timing of a game $\Gamma$ is an assignment of a tuple $(x_v,y_v)$ (two nonnegative real numbers) to each node $v$ such that:
\begin{enumerate}
\item If we label $\Gamma$ with just the $x_v$ values we have a timing of $\Gamma$.
\item If $v$ and $w$ are two nodes belonging to the same player and $v$ is on the path from the root to $w$ then
$l(x_w-x_v)\leq y_w-y_v\leq u(x_w-x_v).$
\end{enumerate}

 An $[l,u]$-timing is a distribution over deterministic
 $[l,u]$-timings. The \emph{timing information} of player $i$ at a node $w$
 given an $[l,u]$-timing consists of the \emph{perceived times}, $y_v$, of all nodes $v$ belonging to that player between the root and $w$. Now an \emph{$(\epsilon,[l,u])$-timing} is an $[l,u]$-timing such that for any two nodes belonging to the same information set, the current player's timing information at the two nodes has total variation distance at most $\epsilon$. An $[l,u]$-timing is an \emph{exact $[l,u]$-timing} if it is a $(0,[l,u])$-timing.
\end{defi}

The next theorem says that even if we can affect the players' clocks by some large constant factor $c$, there still exist games that cannot be $\epsilon$-timed in time $\exp_2^r(\frac{1}{\epsilon})$.

\begin{theo}\label{theo:imperfecttime}
Let $c$ be an integer and let $l,u$ be functions as in Definition \ref{defi:imperfecttime} and such that $l(x)\geq \frac{x}{c}$ and $u(x)\leq cx$. Then for any $r$ there exists a game $\Gamma_{c,r}$ with $16(2c^4+r)+11$ players such that for sufficiently small $\epsilon$ any $(\epsilon,[l(x),u(x)])$-timing of $\Gamma_{c,r}$ has to use time at least $\exp_2^r\left(\frac{1}{\epsilon}\right)$.   
\end{theo}

In order to prove Theorem \ref{theo:imperfecttime} we will show that there exists symmetric choiceless games, where all $(\epsilon,[u,l])$-timing takes a long time to time. To do this, we will use symmetric choiceless games with separators. Unlike for agendas, these separators are not a part of the game, they are only used in the proof. For example, for the game $24|33|1441|22|13$ we could say ``for any $\epsilon$-timing of this game the two middle separator will with high probability either both be much closer to the first than the last separator or both be much closer to the last then the first separator''. This is just a simpler way of saying, ``if $X$ is a timing of $243314412213$ and $(X,X_{|,1},X_{|,2},X_{|,3},X_{|,4})$ is jointly distributed such that $X_{4,1}<X_{|,1}<X_{3,1}$, $X_{3,2}<X_{|,2}<X_{4,2}$, $X_{1,2}<X_{|,3}<X_{2,2}$ and $X_{2,3}<X_{|,4}<X_{1,3}$ then with high probability $X_{|,2}$ and $X_{|,3}$ are either both much closer to $X_{|,1}$ than to $X_{|,4}$ or both much closer to $X_{|,4}$ than to $X_{|,1}$''.

 \begin{lemm}\label{lemm:imperfecttime}
Let $c$ be an integer and let $l,u$ be functions as in Definition \ref{defi:imperfecttime} and such that $l(x)\geq \frac{x}{c}$ and $u(x)\leq cx$. Let $n= 4c^4+1$.

 If $X$ is a symmetric $(\epsilon,[l,u])$-timing of the symmetric choiceless game $\Gamma$ with $2n$ players given by $12\dots n|(n+1)(n+1)(n+2)(n+2)\dots (2n)(2n)||1122\dots nn|(n+1)(n+2)\dots (2n)$, the probability that both
 \begin{enumerate}
 \item $\frac{X_{|,4}-X_{|,2}}{X_{|,4}-X_{|,1}}\geq 2c^{-2}$, and
\item $\frac{X_{|,3}-X_{|,1}}{X_{|,4}-X_{|,1}}\geq 2c^{-2}$
\end{enumerate}
is at most $2\epsilon$, where $X_{|,i}$ denote the actual time of the $j$'th node belonging to player $|$.
 \end{lemm}
\begin{proof}
Suppose we have an $(\epsilon,[l,u])$-timing of $\Gamma$. Let $X_{i,j}$ denote the actual time of the $j$'th node belonging to player $i$, and let $Y_{i,j}$ denote the perceived time. Define $f(x,y,z)=\frac{y-z}{x-z}$ and $f(X_i)=f(X_{i,1},X_{i,2},X_{i,3})$ and similar for $f(Y_i)$. If $x<y<z$ then $f(x,y,z)$ is increasing in $y$ and decreasing in $x$ and $z$. 
%For $i> n$
%\[f(X_i)= \frac{X_{i,2}-X_{i,1}}{X_{i,3}-X_{i,1}}\leq \frac{X_{i,2}-X_{i,1}}{X_{|,4}-X_{|,1}}\]
%and conversely for $i\leq n$
%\[1-f(X_i)=\frac{X_{i,3}-X_{i,2}}{X_{i,3}-X_{i,1}}\leq \frac{X_{i,2}-X_{i,1}}{X_{|,4}-X_{|,1}}.\]  
%As the $[X_{i,1},X_{i,2}]$ for $i>n$ and the $[X_{i,2},X_{i,3}]$ for $i\leq n$ are all pairwise disjoint interval contained in $[X_{|,1},X_{|,4}]$ we get 
%\[\sum_{i=n+1}^{2n}f(X_i)+\sum_{i=1}^n(1-f(X_i))\leq 1.\]

For particular values $x$ and $y$ of $X$ and $Y$, let $n_1$ be the number of $i\leq n$ for which $f(x_{i})\leq 1-\frac{1}{2c^2}$ and let $n_2$ be the number of $i>n$ for which $f(x_{i})\geq \frac{1}{2c^2}$. We will argue that $n_1n_2\leq 4c^4$. For $i>n$ we have
\begin{align*}
f(x_i)=\frac{x_{i,2}-x_{i_1}}{x_{i,3}-x_{i,1}}\leq \frac{x_{i,2}-x_{i,1}}{\sum_{j=1}^n(x_{i,3}-x_{i,2})},
\end{align*}
so if $f(x_i)\geq \frac{1}{2c^2}$ then 
\[\sum_{j=1}^n(x_{i,3}-x_{i,2})\leq 2c^2(x_{i,2}-x_{i,1}).\] 
This hold for each of the $n_2$ values of $i>n$ for which $f(x_{i})\geq \frac{1}{2c^2}$, so
\[2c^2\sum_{j=n+1}^{2n}(x_{i,2}-x_{i,1})\geq n_2\sum_{j=1}^n(x_{i,3}-x_{i,2}).\]
Completely analogously we get
\[2c^2\sum_{j=1}^n(x_{i,3}-x_{i,2}) \geq n_1\sum_{j=n+1}^{2n}(x_{i,2}-x_{i,1}).\]
Putting this together gives
\[4c^4\sum_{j=n+1}^{2n}(x_{i,2}-x_{i,1})\geq 2c^2n_2\sum_{j=1}^n(x_{i,3}-x_{i,2})\geq n_1n_2\sum_{j=n+1}^{2n}(x_{i,2}-x_{i,1}),\]
hence $n_1n_2\leq 4c^4$. 

Furthermore, if $i>n$ and  $f(x_i)\geq  \frac{1}{2c^2}$ we have
\begin{align*}
\frac{1}{2c^2} \leq f(x_i) =\frac{x_{i,2}-x_{i,1}}{x_{i_3}-x_{i,1}} \leq \frac{x_{i,2}-x_{i,1}}{x_{|,4}-x_{|,2}} 
\end{align*}
Thus, if $\frac{x_{|,4}-x_{|,2}}{x_{|,4}-x_{|,1}}\geq 2c^{-2}$ we get 
\[\frac{1}{2c^2} \frac{2}{c^2}\leq  \frac{x_{i,2}-x_{i,1}}{x_{|,4}-x_{|,2}}\frac{x_{|,4}-x_{|,2}}{x_{|,4}-x_{|,1}} \]
so
\[c^4(x_{|,4}-x_{|,1})\geq  x_{i,2}-x_{i,1}\]
As all the intervals $[x_{i,1},x_{i,2}]$ are disjoint and contained in $[X_{|,1},X_{|,4}]$ we see that if $\frac{x_{|,4}-x_{|,2}}{x_{|,4}-x_{|,1}}\geq \frac{2}{c^2}$ we must have $n_1\leq c^4$. Completely analogously we get that if $\frac{x_{|,3}-x_{|,1}}{x_{|,4}-x_{|,1}}\geq \frac{2}{c^2}$ we must have $n_2\leq c^4$. 

For general $i$ we have 
\[1-f(y_i)=\frac{y_{i,3}-y_{i,2}}{y_{i,3}-y_{i,1}}\leq \frac{(x_{i,3}-x_{i,2})c}{(x_{i,3}-x_{i,1})/c}=c^2(1-f(x_i))\]
and $f(y_i)\leq c^2 f(x_i)$. In particular, if $f(x_i)\leq \frac{1}{2c^2}$ then $f(y_i)\leq \frac{1}{2}$ and if $f(x_i)\geq 1-\frac{1}{2c^2}$ then $f(y_i)\geq \frac{1}{2}$. Thus, the number of $i\leq n$ with $f(y_i)< \frac{1}{2}$ is at most $n_1$ and the number of $i>n$ with $f(y_i)\geq \frac{1}{2}$ is at most $n_2$.

In the above $n_1$ and $n_2$ where determined by the value $x$ that $X$ take. Thus, for random $X$ we get random variables $N_1$ and $N_2$. 

By assumption on the timing, for any two $i,i'$ we have $\delta(Y_i,Y_{i'})\leq \epsilon$. In particular if $I$ is uniformly distributed on $[n]$ and $I'$ is uniformly distributed on $[2n]\setminus[n]$ then by Proposition \ref{prop:disjoint} we have $\delta(Y_I,Y_{I'})\leq \epsilon$. Let $g(y_i)$ be the function that is $1$ if $f(y_i)\geq \frac{1}{2}$ and $0$ otherwise. Then we must have
\[\E g(Y_I)-g(Y_{I'})\leq \epsilon.\] 
We know that $\E g(Y_I)\geq \left(\E 1-\frac{N_1}{n}\right)$ and $\E g(Y_{I'})\leq \E \frac{N_2}{n}$. We have $N_1,N_2\leq n$, so if one of them is zero, then $g(Y_I)-g(Y_{I'})\geq 0$. In the cases where both are non-zero we have $N_1N_2\leq 4c^4$ and they are both integers so $N_1+N_2\leq 4c^4+1$. Hence, $g(Y_I)-g(Y_{I'})\geq 1-\frac{N_1}{n}-\frac{n_2}{n}\geq \frac{n-4c^4-1}{n}\geq 0$, so for all values of $X$ we have $g(Y_I)-g(Y_{I'})\geq 0$. We have shown that when $\frac{x_{|,4}-x_{|,2}}{x_{|,4}-x_{|,1}}\geq \frac{2}{c^2}$ then $n_1\leq c^4$ and when $\frac{x_{|,3}-x_{|,1}}{x_{|,4}-x_{|,1}}\geq \frac{2}{c^2}$ we must have $n_2\leq c^4$. Thus, this case contributes with $g(Y_I)-g(Y_{I'})\geq 1-\frac{N_1}{n}-\frac{N_2}{n}\geq 1-\frac{2c^4}{n}\geq \frac{1}{2}$, so the probability that $\frac{x_{|,4}-x_{|,2}}{x_{|,4}-x_{|,1}}\geq \frac{2}{c^2}$ and $\frac{x_{|,3}-x_{|,1}}{x_{|,4}-x_{|,1}}\geq \frac{2}{c^2}$ is at most $2\epsilon$. 
\end{proof}

\begin{proof}[Proof sketch of Theorem \ref{theo:imperfecttime}]
Same idea as the induction step of the proof Theorem \ref{theo:lower}: We will construct symmetric choices games that takes a long time to $(\epsilon,[u,l])$-time. An argument analogously to the proof of Proposition \ref{prop:sym} shows that it is enough to consider symmetric timings. By Corollary \ref{coro:lowerbound} there exists a game $\Gamma_{r}$ that cannot be $\epsilon$-timing in time less than $\exp_2^{r}(\epsilon^{-1})$. We then modify $\Gamma_{r}$: Between any two nodes of $\Gamma_{r}$ we put three separators, and before the first node and after the last, first three separators, then outside that we put one node for each player, and then a ``$|$'' before and after all the other nodes. For example, the game $233112$ becomes
\[|123|||2|||3|||3|||1|||1|||2|||123|.\]
 Around the first $4$ separators, we put $2(4k^4+1)$ players as in Lemma \ref{lemm:imperfecttime}, around separator $2$ to $5$ we put $2(4k^4+1)$ other players and so on. When we get to separators $5$ to $8$ we can reuse players as in the induction step of the proof of Theorem \ref{theo:lower}, so in total we use $8(4k^4+1)$ players, and these ensure that with probability $O(\epsilon)$ gap between the separators will either increase fast or decrease fast. The players from $\Gamma_{r}$ will now be able to see if the gaps are decreasing or increasing. If they are increasing the players will look at the difference between their first node and their $i$'th. Because their perception of this gap is of by at most a factor $k$, and the distances between these gaps will differ by almost a factor $k^4$, the ordering of the perceptions of these distances will be the same at the actual ordering. Hence, the logarithms of the perceived distances is a timing of $\Gamma_{r}$ and must use numbers as large as $\exp^{r}_2(\Omega(\frac{1}{\epsilon}))$, in particular any $(\epsilon,[l,u])$-timing of $\Gamma_{r,c}$ must use time at least $\exp^{r+1}_2(\Omega(\frac{1}{\epsilon}))$ and hence $\exp_2^r(\frac{1}{\epsilon})$ for sufficiently small $\epsilon$. 
\end{proof}

The next theorem shows that the above is the strongest theorem we can hope for: if we can make the players' clocks go faster or slower by more that a constant factor, we can implement all games.

\begin{theo}\label{theo:imperfectconstruction}
Let $\Gamma$ be a game and $l,u$ functions as in Definition \ref{defi:imperfecttime} with $\frac{u(t)}{l(t)}\to \infty$ as $t\to \infty$. Then $\Gamma$ is exactly $[l,u]$-timeable.
\end{theo}
\begin{proof}
Let $M$ be the maximal number of nodes in any history of $\Gamma$. As $\frac{u(t)}{l(t)}\to \infty$ we must have either $\limsup_{t\to \infty} \frac{u(t)}{t}=\infty$ or $\limsup_{t\to\infty}\frac{t}{l(t)}=\infty$. In the first case we can find $t_0\geq 1$ such that $u(t_0)\geq M^2t_0 \geq M l(Mt_0)$. Similarly, in the second case we can find $t_0\geq 1$ such that $M^2 l(Mt_0)\leq Mt_0\leq Mu(t_0)$, hence $Ml(Mt_0)\leq u(t_0)$.    

Now we define an $[l,u]$-timing of $\Gamma$. If a node $v$ has distance $i_v$ from the root, we set $x_v=i_vt_0$, and if $v$ is the $j_v$'th node belonging that player in the history leading to $v$, then $y_v=j_vl(Mt_0)$. Clearly, the $x_v$ give a deterministic timing of $\Gamma$ as $t_0\geq 1$. If $v,w$ are two nodes belonging to the same player with $v$ on the path from the root to $w$ we have $y_w-y_v=(j_w-j_v)l(Mt_0)$ and
\[l(x_w-x_v)\leq l(Mt_0)\leq (j_w-j_v)l(Mt_0)\leq M l(Mt_0)\leq u(t_0)\leq u(x_w-x_v).\]
Thus, we have constructed an $[l,u]$-timing. If $v$ and $w$ a two nodes belonging to the same information set, we must have $j_v=j_w$, and the timing information for each of them is $(l(Mt_0),2l(Mt_0),\dots,j_ul(Mt_0))$, so it is an exact $[l,u]$-timing.
\end{proof}

\section{Conclusion}

Not every extensive-form game can be naturally implemented in the
world.  Games with imperfect recall constitute a well known example of
this.  In this paper, we have drawn attention to another feature that
is likely to prevent the direct implementation of the game in the
world: games that are not exactly timeable.  We gave necessary and
sufficient conditions for a game to be exactly timeable and showed
that they are easy to check.  Most of the technical contribution
concerned approximately timing games; we showed that this can always
be done, but can require large amounts of time.

Future research can take a number of directions.  Does restricting
attention to exactly timeable games allow one to prove new results about
these games, or develop new algorithms for solving them---as is the case
for perfect recall?  It is conceivable that the possibility of games that
are not exactly timeable has been an unappreciated and unnecessary
roadblock to the development of certain theoretical or algorithmic results.
Can our techniques be applied to the design of protocols that should not
leak information to participants by means of the time at which they receive
messages? Are there natural families of games for which we can obtain desirable
bounds for the amount of time required to approximately time them?

\section*{Acknowledgements}

We thank J\"orgen Weibull for helpful feedback. 
Conitzer thanks ARO and NSF for support under grants W911NF-12-1-0550,
W911NF-11-1-0332, IIS-0953756, CCF-1101659, and CCF-1337215.

  \bibliographystyle{plain} 
\bibliography{timeabilityArXiv}

	 \end{document}